\documentclass[12pt]{article}

\usepackage{graphicx}
\usepackage{natbib}
\usepackage{float}
\usepackage{doi}
\usepackage[caption=false]{subfig}

\usepackage{arxiv}

\usepackage[utf8]{inputenc} 
\usepackage[T1]{fontenc}    
\usepackage{hyperref}       

\hypersetup{
  colorlinks   = true, 
  urlcolor     = blue, 
  linkcolor    = blue, 
  citecolor   = blue 
}

\usepackage{amssymb}
\usepackage{amsmath,amsfonts}

\newtheorem{definition}{Definition}
\newtheorem{proposition}{Proposition}

\newtheorem{proof}{Proof}

\title{The GLD-plot: A depth-based plot to investigate unimodality of directional data}

\author{ 
Giuseppe~Pandolfo \\
	Department of Industrial Engineering\\
	University of Naples Federico II\\
	Naples, Italy\\
	\texttt{giuseppe.pandolfo@unina.it} 
}

\date{}

\renewcommand{\headeright}{ }
\renewcommand{\undertitle}{ }
\renewcommand{\shorttitle}{The GLD-plot: A depth-based plot to investigate unimodality of directional data}

\begin{document}
\maketitle

\begin{abstract}
A graphical tool for investigating unimodality of hyperspherical data is proposed. It is based on the notion of statistical data depth function for directional data which extends the univariate concept of rank. Firstly a local version of distance-based depths for directional data based on aims at analyzing the local structure of hyperspherical data is proposed. Then such notion is compared to the global version of data depth by means of a two-dimensional scatterplot, i.e. the GLD-plot. The proposal is illustrated on simulated and real data examples.
\end{abstract}

\section{Introduction}
\label{intro}

Multivariate analysis plays a fundamental role in statistics, and nowadays most statistical experiments are multivariate by nature, and large scale multivariate data sets are now tractable thanks to the advances in technology. However, classical multivariate analysis relies on some distributional assumptions that are often difficult to justify in real life applications.

This work aims at introducing a general nonparametric graphical tool based on the concept of data depth functions for directional data. The proposed methodology provides a nonparametric approach for graphically detecting multivariate distributional unimodality on hyperspheres. Specifically, information about unimodality of hyperspherical distributions are obtained through data depths, and they can be displayed and visualized in a simple two-dimensional plot. Such graph is based on an analysis of the rankings derived from a data depth function and its local counterpart, so that they can offer an easy interpretable picture of the distributions. Indeed, the proposed depth-based tool can be interpreted as a multivariate generalization of standard univariate rank methods with the important difference that the ranking in the univariate case is linear (going from the smallest to the largest), while the multivariate ranking considered here is a \textit{center-outward ranking} induced by depth functions.

The notion of centrality also plays an important role in the domain of graph analytics, and the use of depth-induced rankings to investigate distributional features has been already used for analyzing data in $\mathbb{R}^{q}$ by means of graphical tools. \cite{liu1999multivariate} proposed the ``sunburst plot'' as a bivariate generalization of the box-plot and the DD-(depth versus depth) plots. In the same year, \cite{rousseeuw1999bagplot}, proposed the bagplot, a bivariate generalization of the univariate boxplot by exploiting the notion of halfspace location depth. Later on, \cite{li2012dd} used the DD-plot to perform  classification of data in $\mathbb{R}^{q}$. A nonparametric classification procedure based on the DD-plot was introduced also by \cite{lange2014fast}, while a clustering algorithm based on it was proposed by \cite{singh2016nonparametric}. The concept of data depth was also used to build control charts for monitoring processes of multivariate quality measurement \citep[see, e.g.][]{liu1995control,bae2016data}. 

However, despite the increasing interest for the analysis of data in $\mathbb{R}^{q}$, the adoption of depth-based visualization techniques for the analysis of directional data has been neglected so far, except for a recent work about the classification of data on the unit circle through the DD-plot \citep{pandolfo2018note}. 
  
The proposed method is motivated by means of simulated and real data examples. The remainder of the article is organized as follows. In Section \ref{sec:background} some background about directional data and the concept of data depth function is presented. In Section \ref{sec:locdep} the definition of a local notion of distance-based depth for directional distribution is proposed. A depth-based exploratory graphical tool for investigating hyperspherical unimodality is introduced in Section \ref{sec:Applications}. Two real data examples are offered in Section \ref{sec:RealDataExample}, while some final remarks are reported in Section \ref{sec:conc}.

\section{Background: Directional data and data depth concept}
\label{sec:background}

A directional data sample is a collection of $n$ observations lying on the surface of the unit $\left(q-1\right)-$dimensional hypersphere $S^{q-1}:=\left\{x \in \mathbb{R}^{q}:x'x=1\right\}$ of $\mathbb{R}^{q}$, for $q \geq 2$. Each observation is recorded as a direction of a unit vector in $\mathbb{R}^{q}$. Such data arise in many scientific fields such as geology, meteorology and biology, just to name a few.
Many applications can be found for $q = 2$ (circular data) or $q = 3$ (spherical data), but also in higher dimensions, e.g. in gene-expression analysis \citep[see ][]{banerjee2004}, text mining \citep[see ][]{banerjee2005,hornik2014}, or pattern recognition \citep[see][]{wilson2014}.

Data depth function is an important nonparametric tool for the analysis of complex data such as functional and directional data. It provides a center-outward ordering of the data and leads to a ranking of data which can be exploited for describing different features of the data distribution.  

Consider a bounded spherical distance $d(\cdot, \cdot)$ whose upper bound between any two points on $S^{q-1}$ is denoted as $d^{\sup}:=\sup \{d(x, y): x,y\in S^{q-1}\}$. Then, let $x$ be a point on $S^{q-1}$ and $F$ a distribution on $S^{q-1}$.
\begin{definition}[Directional distance-based depths] 
\label{defclass}
A directional distance-based depth of~$x \in S^{q-1}$ with respect to $F$ is defined as
\begin{equation}
\label{eq:class}
D_{d}\left(x; F\right) := d^{\sup} - E\left[d\left(x; X_{1}, \ldots, X_{n}\right)\right],	
\end{equation}
where~$X_{1}, \ldots, X_{n}$ is a random sample from $F$.
\end{definition}

According to \cite{pandolfo2018distance}, a distance-based statistical depth function for directional data should satisfy the four desirable properties listed below. 
\begin{itemize}
	\item[\textbf{P1.}] \textbf{Rotational invariance:} $D_{d}\left(x,F\right) = D_{d}\left(Ox,OF\right)$ for any $q \times q$ orthogonal matrix $O$;
	\item[\textbf{P2.}] \textbf{Maximality at center:} $\underset{x \in S^{q-1}}{\sup}D_{sph}\left(x^{*},F\right)$ for any distribution $F$ with center at $x^{*}$;
\item[\textbf{P3.}] \textbf{Monotonicity relative to the deepest point:} $D_{d}\left(x,F\right)$ decreases along any geodesic path $t \mapsto x_{t}$ from the deepest point to its antipodal point $-x^{*}$, i.e. \linebreak $D_{d}\left(x\left(t_{1},F\right)\right) \leq D_{d}\left(x\left(t_{2},F\right)\right)$ for $\left(t_{1},F\right) > \left(t_{2},F\right)$;
\item[\textbf{P4.}] \textbf{Minimality at the antipode to the center:} $D_{d}\left(-x^{*},F\right) = \underset{x \in S^{q-1}}{\inf}{d}\left(x^{*},F\right)$ for any $F$ with center at $x^{*}$.
\end{itemize}
P2 and P3 are the extension of properties 2 and 3 of depth functions for data in $\mathbb{R}^{q}$ as defined in \cite{zuo2000}. Conversely, P1 and P4 characterize a distance-based depth for directional data. 

Depths are known to be monotonically decreasing from the deepest point (i.e. the point at which they are maximized), and recognize just one center regardless of the number of the underlying modes. Consequently, while they are well suited for unimodal distributions, they are not able to capture the multimodality of a distribution. Hence, depths substantially fail here in describing the structure of the data. 

For this reason, over the last few years, some definitions of local depth functions were introduced with the purpose of revealing local characteristics of the data, and have already found some interesting applications. The idea which lies behind this approach is to compute the depth of a point with respect to a bounded neighborhood of it. As a consequence, local depth functions admit multiple maxima that correspond to local centers. The first definition of local depth was introduced by \cite{agostinelli2011} who exploited the local space geometry to provide the notions of local simplicial and halfspace depths. Their definition relies on a tuning parameter to define a constant size neighborhood of each point. However, the main drawback of such notions of depth (which are based on geometrical structures) regards their high computational cost, especially when dimensions and sample size increase. Later, \cite{paindaveine2013} proposed a method aimed at converting global depths into local ones by conditioning the distribution to appropriate depth-based neighborhoods. 

Hence, the depth problem can be dealt with a \textit{global} or a \textit{local} approach also in directional data analysis allowing for the possibility of getting different insights into specific scenarios. However, while a vast body of literature exists on local depths for multivariate and functional data, comparatively only the definition of local simplicial depth for (hyper)spherical data have been proposed \citep[see ][]{agostinelli2012}. 

\section{Local distance-based depths}
\label{sec:locdep}

In this section a general definition of local depths is proposed by generalizing the definition of distance-based depth functions for directional data introduced by Pandolfo et al. \cite{pandolfo2018distance}. 
The proposed local depth notion is based on the idea of getting insights into the shape of a distribution $F$ by restricting a global distance-based depth measure to a neighborhood of each point on the spherical space. Hence, for a given point $\mathit{x}$ on $S^{q-1}$, all the points within a certain distance $\mathit{\delta}$ to $\mathit{x}$ are considered, i.e. the set of points which are within a certain degree of \textit{closeness} to $\mathit{x}$ which defines a spherical cap ``centered'' at $\mathit{x}$. 

\begin{definition}[Local distance-based depths for directional data] Consider any \linebreak bounded spherical distance function $d$ and a distribution $F$ on $S^{q-1}$. Then, for any positive value $0 < \mathit{\delta} < d^{\sup}$, the local directional distance-based depth of a point $\mathit{x} \in S^{q-1}$ with respect to $F$ is defined as 
$$
LD^{\mathit{\delta}}_{d}\left(x; F\right) := d^{\sup} - E_{F}\left[d\left(x; X_{1}, \ldots, X_{n}\right)|d\left(x; X_{1}, \ldots, X_{n}\right) \leq \delta \right],
$$ 
where $X_{1}, \ldots, X_{n}$ is a random sample from $F$.
\label{deflocal}
\end{definition}

The parameter $\delta$ determines how central the point $x$ is within the space conditional to a given window around $x$. 
The value of the local depth function is dependent on the properties of the distribution $F$ within the neighborhood of $\mathit{x}$ defined by the parameter $\mathit{\delta}$. Specifically, in the case of (hyper)spherical data, the distribution of the distance within the considered neighborhood can help to provide information on the data structure. 

The value of the parameter $\mathit{\delta}$ defines to what extent the depth function is local. When it approaches the supremum of the chosen spherical distance, the local depth approaches the global one, and a unique center will be defined. As $\mathit{\delta}$ gets smaller, local depth allows for the detection of local points of maximum and minimum.
Hence the choice of the parameter $\mathit{\delta}$ crucially depends on the application at hand and it does not exist a best strategy in general. Thus, it is advisable, in every specific application, to define appropriate strategies.  
However, one note of caution is due. Directional data on $S^{q-1}$ wrap around since they have finite length in space, and this must be taken into account when defining the value of $\mathit{\delta}$.

Note that when the distribution is continuous and unimodal, local and global behave very similarly and provide similar rankings. 

Since $E_{F}\left[d\left(x; X_{1}, \ldots, X_{n}\right)|d\left(x; X_{1}, \ldots, X_{n}\right) \leq \mathit{\delta} \right]$ is increasing on $\left[0, d^{\sup}\right)$ as $\mathit{\delta}$ increases, when $\mathit{\delta}$ becomes smaller the local depth values become larger, i.e.
$$
LD^{\mathit{\delta_{1}}}_{d}(x; F) \geq LD^{\mathit{\delta}_{2}}_{d}(x;F) \quad \text{for} \quad  0 < \mathit{\delta}_{1} < \mathit{\delta}_{2}.
$$

\noindent In addition, it is trivially satisfied that $\underset{\delta \rightarrow 0^{+}}{\lim} LD^{\mathit{\delta}}_{d}(x;F) = d^{\sup}$ which implies 

\begin{equation*}
0 \leq LD^{\mathit{delta}}_{d}\left(x; F\right) - D_{d}\left(x, F\right) < d^{\sup} - D_{d}\left(x, F\right)
\label{rem1}
\end{equation*}	

The relationship between the local and the global distance-based depth function is explained in the following proposition.
\begin{proposition}
\label{prop2}
\noindent Let $F$ be an absolutely continuous distribution on $S^{q-1}$. Then, for a point $x \in S^{q-1}$ and $0 < \delta < d^{\sup}$ we have
	$$D_{d}\left(x; F\right) = LD^{\mathit{\delta}}_{d}\left(x; F\right) - \left\{E_{F}\left[d\left(x; F\right)\right] - E_{F}\left[d\left(x; F\right)|d\left(x; F \right) \leq \delta \right]\right\}.$$
\end{proposition}

\begin{proof}
The proof is straightforward and thus omitted. 
\end{proof}

\begin{proposition}
\label{thm1}
$LD_{d}\left(x, F\right)$ is rotational invariant using any rotational invariant spherical distance function $d$ on $S^{q-1}$.
\end{proposition}

\begin{proof}
Clearly the result follows if $d$ is rotational invariant of the form $d\left(Ox, Oy\right) = d\left(x, y\right)$ for any two points $x$ and $y$ on $S^{q-1}$, and for any $d \times d$ orthogonal matrix $O$ for any subset $\mathcal{N} \subset S^{q-1}$.
\end{proof}

As pointed out by \cite{agostinelli2011}, local depth functions cannot be strictly considered depth functions (this holds true for data in $\mathbb{R}^{d}$ and on $S^{q-1}$ as well). This is because they may not possess the monotonicity property given that they are aimed at searching for more sensitivity to local features of the distribution.

For the purpose of this work three notions of distance-based depths for directional data, namely the arc distance depth ($D_{\rm arc}$), the cosine distance depth ($D_{\cos}$) and the chord distance depth ($D_{\rm chord}$) are considered. They are based on the following distance measures:
\begin{itemize}
\item the \emph{arc length distance} $d_{\rm arc}$, associated with $\eta(t)=\eta_{\rm arc}\left(t\right) = \arccos\left(t\right)$;
\item the \emph{cosine distance} $d_{\cos}$, associated with $\eta\left(t\right) = \eta_{\cos}\left(t\right)=1-t$;
\item \sloppy the \emph{chord distance} $d_{\rm chord}$, defined as $d_{\rm chord}(x,y)=\|x-y\|=\sqrt{2(1-x'y)} =: \delta_{\rm chord}(x'y)$.
\end{itemize}
Note that $0 \leq d_{\rm arc} \leq \pi$, $0 \leq d_{\cos} \leq 2$ and $0 \leq d_{\rm chord} \leq 2$.\\

The local version of depths based on the above listed distances can be derived according to Definition \ref{deflocal}. Henceforth they will be denoted as $LD_{\rm arc}$, $LD_{\cos}$ and $LD_{\rm chord}$. To illustrate how they work in capturing the local features of a reference distribution, their empirical behavior is investigated by means of simulated data sets. Specifically, the focus is put on the question of whether the shapes of the local distance-based depth functions reflect the characteristics of the underlying distribution.

For this purpose the von Mises-Fisher distribution, usually denoted by $vMF\left(\mu, \kappa\right)$, is considered. This is the directional analog of the Gaussian distribution in $\mathbb{R}^{q}$ and has probability density function given by
\begin{equation}
f\left(x,\mu,\kappa\right) = \frac{\kappa^{q/2-1}}{2\pi^{q/2} I_{q/2-1}\left(\kappa\right)}\exp\left\{\kappa\cos\left(x' \mu\right)\right\},
\label{eq:vMF}
\end{equation}
where $I_{v}\left(\kappa\right)$ is the modified Bessel function of first kind and order $v$. The density is parametrized by the mean direction $\mu$ and the concentration parameter $\kappa$, that measures how strongly data are concentrated around the mean direction $\mu$ (larger values indicate stronger concentration of unit vectors around $\mu$).

Without loss of generality, for the sake of illustration let us considering only simulated data restricted to the circular case ($q = 2$). Firstly, two samples drawn from a unimodal von Mises-Fisher distribution and an equally-weighted mixture of von Mises-Fisher distributions on $S^{1}$ are considered, that is $F_{1} = vMF\left(\pi, 2\right)$ and $F_{2} = \frac{1}{2}vMF\left(\pi, 0.5\right) + \frac{1}{2}vMF\left(3\pi/2, 5\right)$, respectively. 

Two additional scenarios were considered, that is a bimodal distribution with the two modes which are $7/12\pi$ ($105^{\circ}$) far from each other and on the more interesting case of symmetric bimodal antipodal distribution, that is the underling density $f$ is such that $f\left(x\right) = f\left(-x\right)$ for any $x \in S^{q-1}$, where global depths usually are constant and thus any provided order appears to be unrealistic (see \citealp{pandolfo2018distance}, \citealp{liu1992} and \citealp{ley2014}).

The parameter $\mathit{\delta}$ was held constant all over the space for all cases for all the considered depth functions, i.e. $\mathit{\delta}_{\rm arc} = \pi/2$, $\mathit{\delta}_{\cos} = 1$ and $\mathit{\delta}_{\rm chord} = \sqrt{2}$. This is to consider, for the computation of each depth, a neighborhood of a given point $\mathit{x}$ such that the angle between $\mathit{x}$ and any other point on the boundary of the neighborhood was equal to $90^{\circ}$. 
The plots of the corresponding local distance-based depth functions $LD_{\rm arc}$, $LD_{\cos}$ and $LD_{\rm chord}$, along with the density curve of the considered distributions are reported in Figures \ref{fig:fig1} and \ref{fig:fig2}. In the case of unimodal distribution the local depths reflect the distributional pattern. When the distribution is asymmetric (generated according to a mixture of two von Mises-Fisher distributions), local depths appear to be able to detect the mode and the bump around the minor mode of the distribution although it is less visible for the $LD_{\cos}$. For the bimodal distributions, local depth functions clearly show two local maxima located about the two modes of the distribution.
\begin{figure}[h!]
  \begin{minipage}[b]{0.5\linewidth}
    \centering
		\includegraphics[width=1\textwidth]{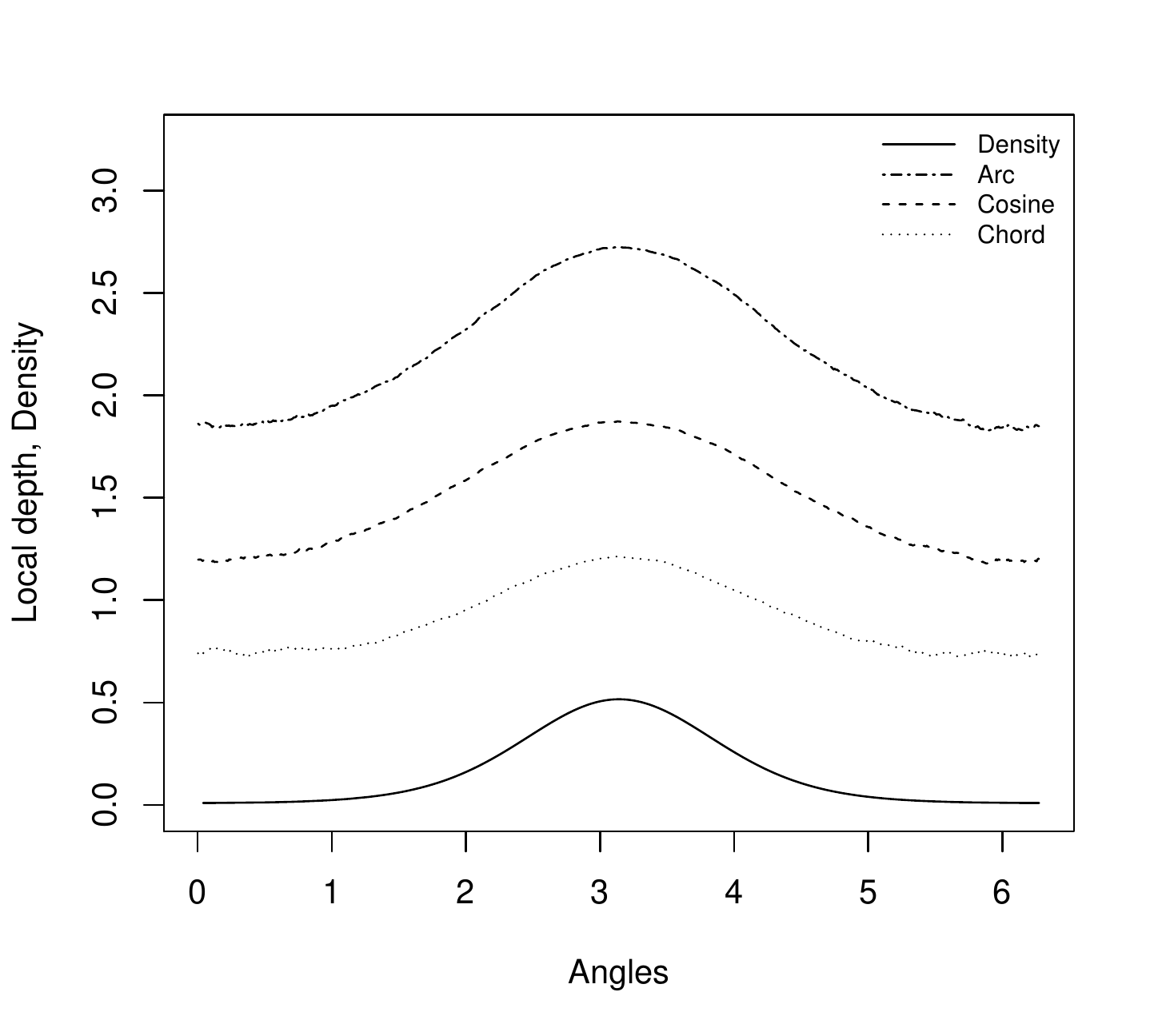}
  \end{minipage}
  \hspace{0.15cm}
  \begin{minipage}[b]{0.5\linewidth}
    \centering
		\includegraphics[width=1\textwidth]{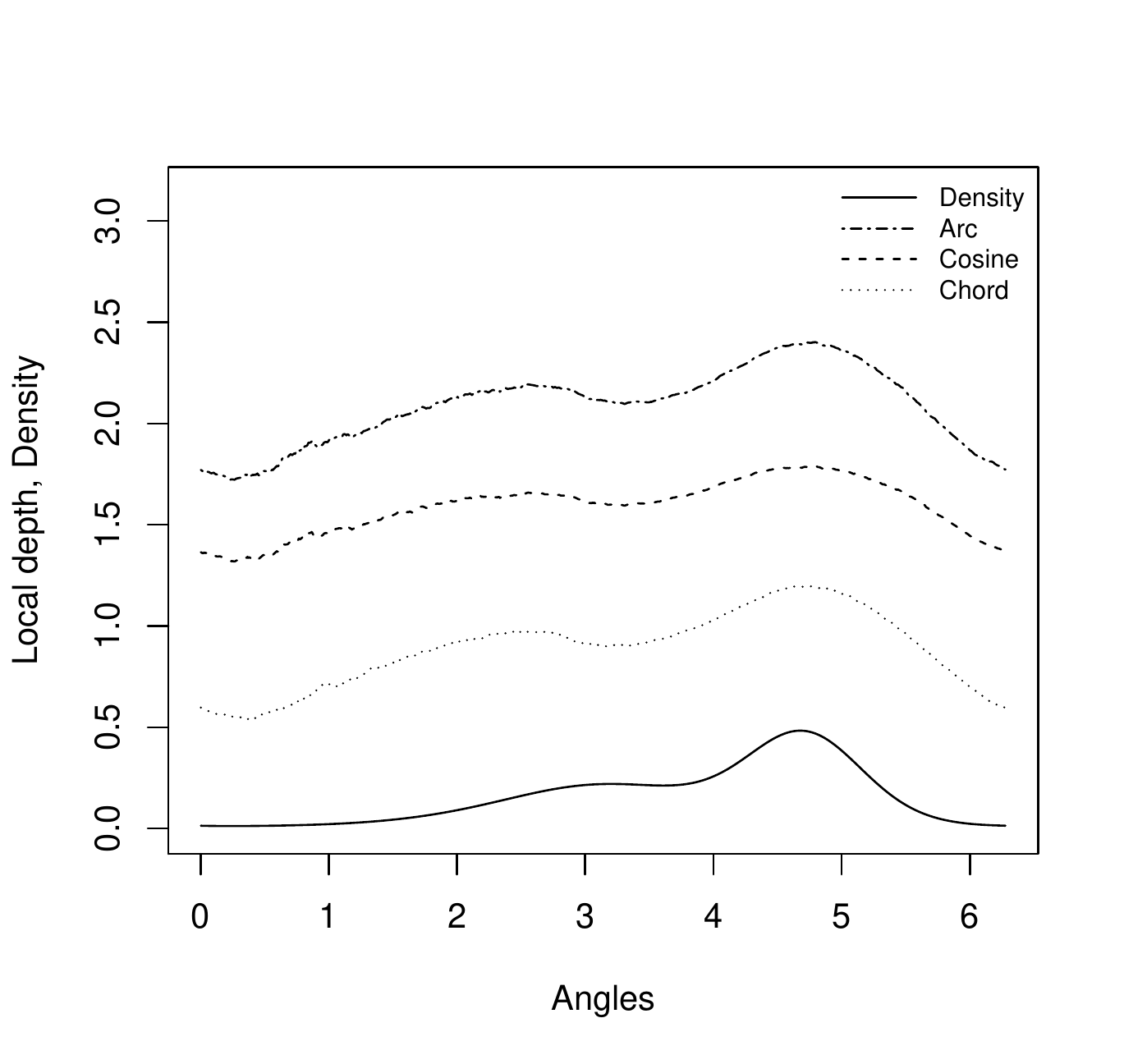}
  \end{minipage}
	    \caption{Plots of local arc, cosine and chord distance depth functions for a circular uniform (left panel) and mixture of von Mises-Fisher distributions (right panel), along with the corresponding density curve.}
			    \label{fig:fig1}
\end{figure}

\begin{figure}[h!]
  \begin{minipage}[b]{0.5\linewidth}
    \centering
		\includegraphics[width=1\textwidth]{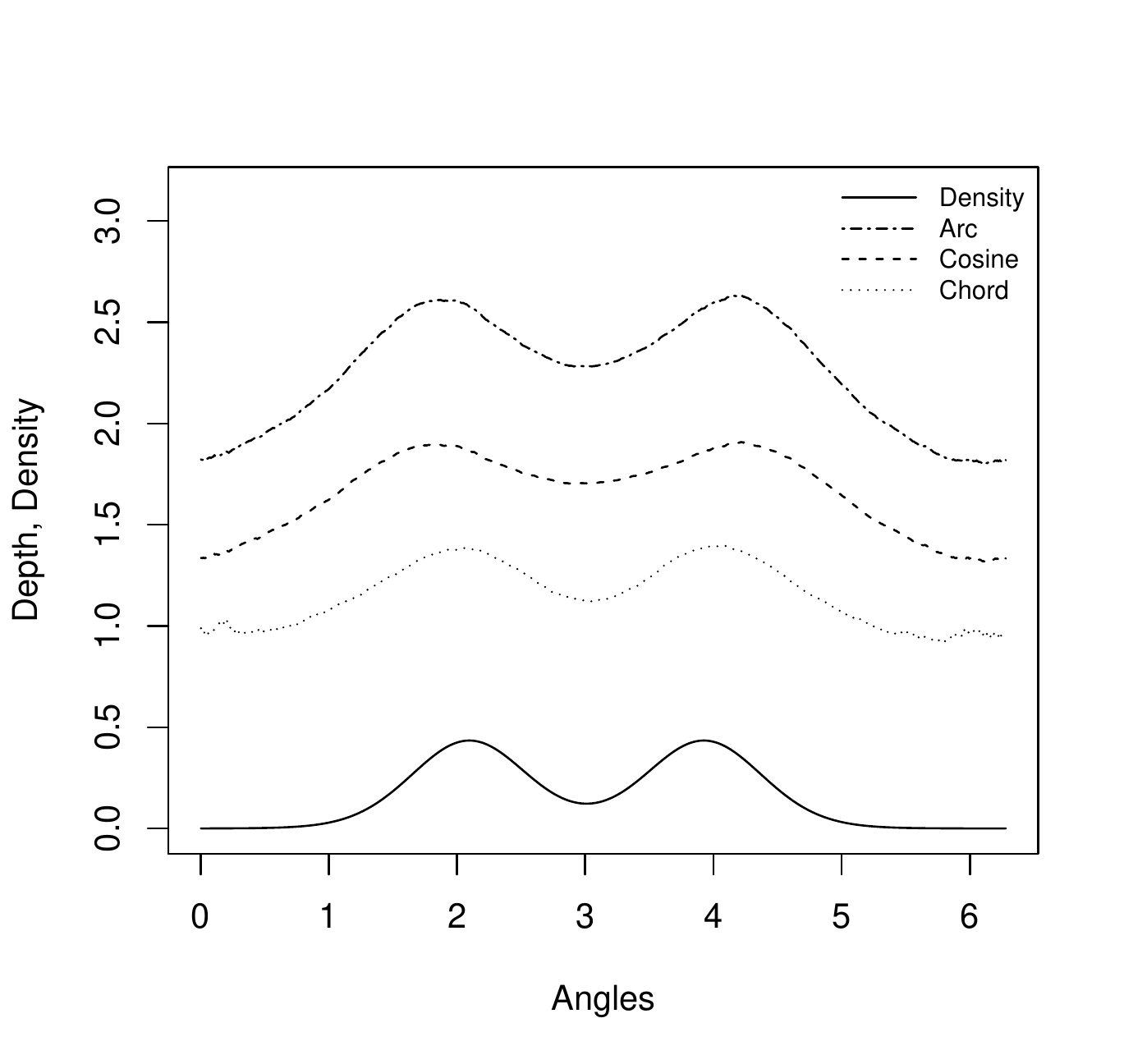}
  \end{minipage}
  \hspace{0.15cm}
  \begin{minipage}[b]{0.5\linewidth}
    \centering
		\includegraphics[width=1\textwidth]{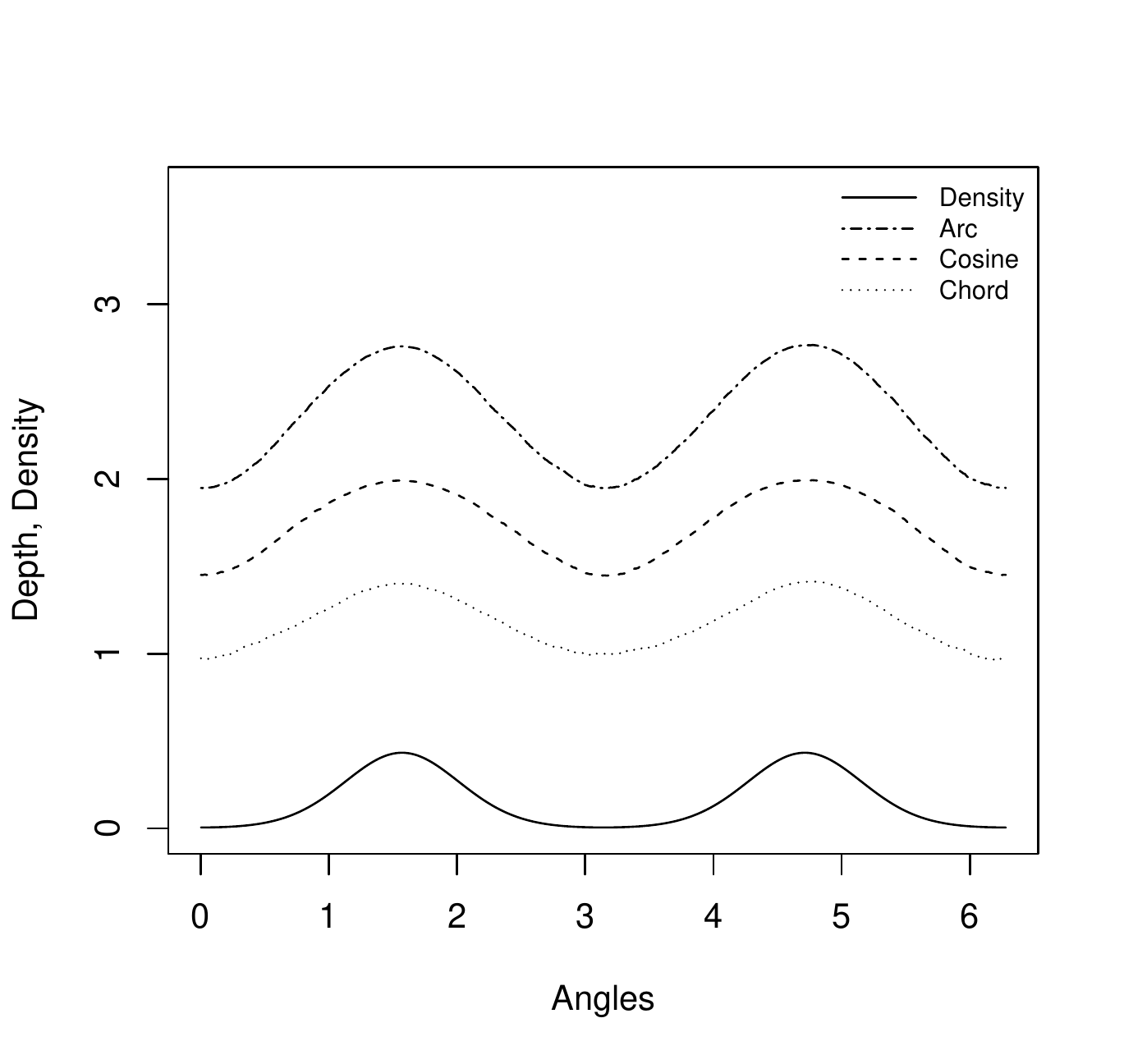}
  \end{minipage}
	    \caption{Plots of local arc, cosine and chord distance depth functions for a bimodal circular distribution with modes $\frac{7}{12}\pi$ ($105^{\circ}$) far from each other (left panel) and an antipodal distribution (left panel), along with the corresponding density curve.}
			    \label{fig:fig2}
\end{figure}

\section{An exploratory graphical tool for investigating hyperspherical unimodality}
\label{sec:Applications}

In this section, a depth-based tool for the graphical exploratory analysis of data on hyperspheres is introduced. Specifically, it aims at exploring hyperspherical unimodality.

\subsection{The Global-Local Depth plot}
\label{sec:exploratory}

Here, by combining the properties of the global and local depth functions, a straightforward plot to investigate unimodality of directional data is proposed. The main idea relies on the fact that if a distribution is unimodal, its local depth function approaches the corresponding global depth. Conversely, when a distribution is not unimodal, the rankings provided by the two functions differ. The farther the data are from unimodality, the larger the difference between the global and local rankings. In practice, the global and the local depth values of a point can be visually compared by means of a two-dimensional scatterplot where the $x$-coordinates are the global depth of the corresponding data point and the $y$-coordinates are the local depth of the corresponding data point. To simplify the comparison, the values of the sample depths are normalized by scaling them between 0 and 1. For this reason, the plot is called Global-Local Depth plot, GLD-plot hereafter. If the distribution is unimodal and thus the set of the deeper local points does not substantially differ from the corresponding set of the deeper global depth points, the plot will exhibit a concentration on the upper-right corner. In case of strong unimodality, the ranks of the two depth functions will coincide, and points on the plot will roughly form a straight diagonal line. On the other hand, departure from unimodality will show different scenarios, obviously depending on the kind of departure.
To facilitate the user to interpret the visual result of a GLD-plot, two dashed lines are added so that the plot is partitioned in four sections. The lines select the $50\%$ deepest sample points according to the global and local functions, respectively. Roughly, if points in the GLD-plot substantially lie within the upper-right and bottom-left quadrants, this will suggest that the distribution is substantially unimodal.  
In addition, an upward line whose slope is given by the ratio between the $50$th percentile of the local depth and the $50$-th percentile of the global depth is plotted. In case of strong unimodality and for a ``medium'' value of $\mathit{\delta}$, such line becomes a 45-degree line and points will lie on it.

To illustrate, data according to different cases were simulated and the corresponding GLD-plots were depicted. Only results coming from the adoption of $D_{\cos}$ and $LD_{\cos}$ are reported. Very similar results were obtained with $D_{\rm arc}$ and $LD_{\rm arc}$, and $D_{\rm chord}$ and $LD_{\rm chord}$ for the here considered setups. However, it is worth underlying that different notions may produce slightly different views especially in case of non unimodal distributions.

In order to explore the data, the parameter $\mathit{\delta}_{\cos}$ was set equal to $0.29$, $0.5$ and $1$. Such values were chosen in order to consider a neighborhood of a given point $\mathit{x}$ such that the angle between $\mathit{x}$ and any other point on the boundary of the neighborhood was equal to $45^{\circ}$, $60^{\circ}$ and $90^{\circ}$, respectively.

First, for $q = 5$ two unimodal samples were generated according to a von Mises-Fisher distribution with location parameter $\mu=(1,0,0,0,0)$ and concentration parameter $\kappa=5$ and $20$. The sample size was set equal to $250$. The corresponding GLD-plots are depicted in Figures \ref{firstgld} and \ref{secondgld}.

When the data are less concentrated around the mean direction (i.e. $\kappa = 5$), the GLD-plots show many points on the upper-right corner, indicating the $50\%$ of the highest values of both local and global depths are reached approximately in the same area of the sample space. When $\mathit{\delta}_{\cos} = 1$, one can see that points do not deviate too much from the straight line, while a greater difference can be seen for $\mathit{\delta}_{\cos} = 0.5$ even though unimodality seems still holding. When an even smaller value of $\mathit{\delta}_{\cos}$, i.e. $0.29$, the entire data cloud lies at the bottom of the plot because of smaller values of the local depth. Here, despite some more points fall in the upper-left and bottom-right quadrants with respect to the other two GLD-plots, the majority of the data lie again within the upper-right and bottom-left quadrants.

\begin{figure}[ht]%
 \centering
 \subfloat[]{\includegraphics[width=0.40\textwidth]{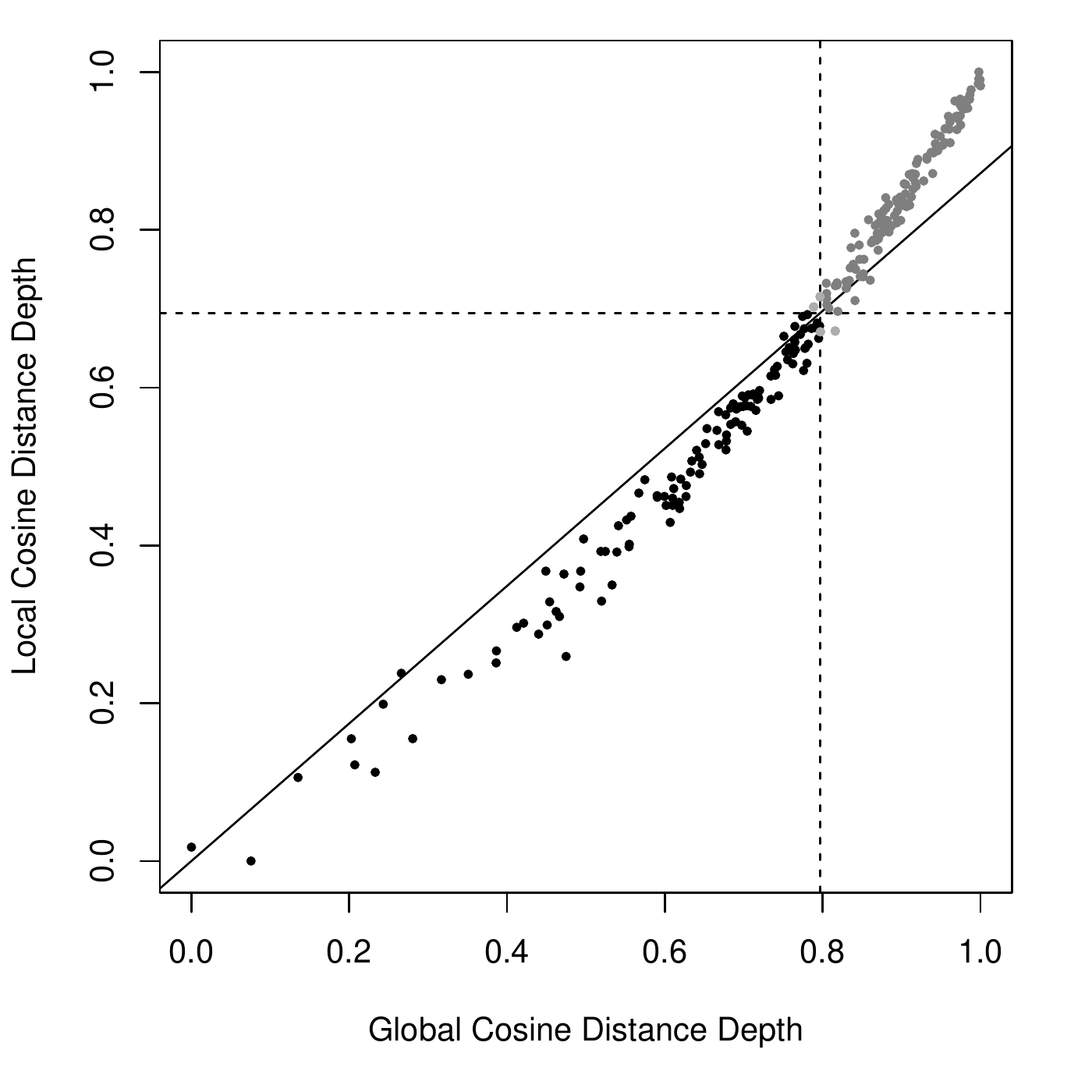}\label{fig:a}}%
 \subfloat[]{\includegraphics[width=0.40\textwidth]{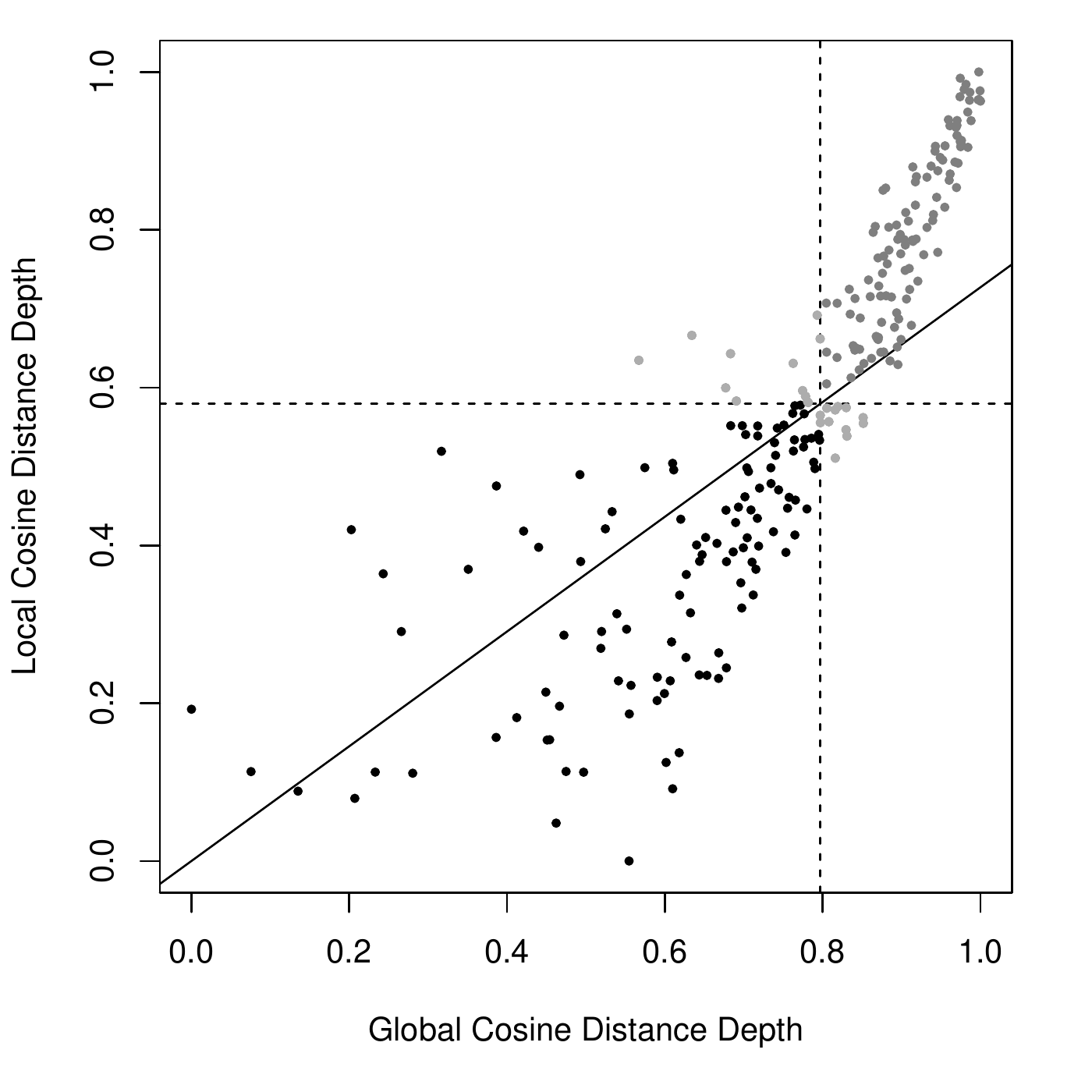}\label{fig:b}}\\
 \subfloat[]{\includegraphics[width=0.40\textwidth]{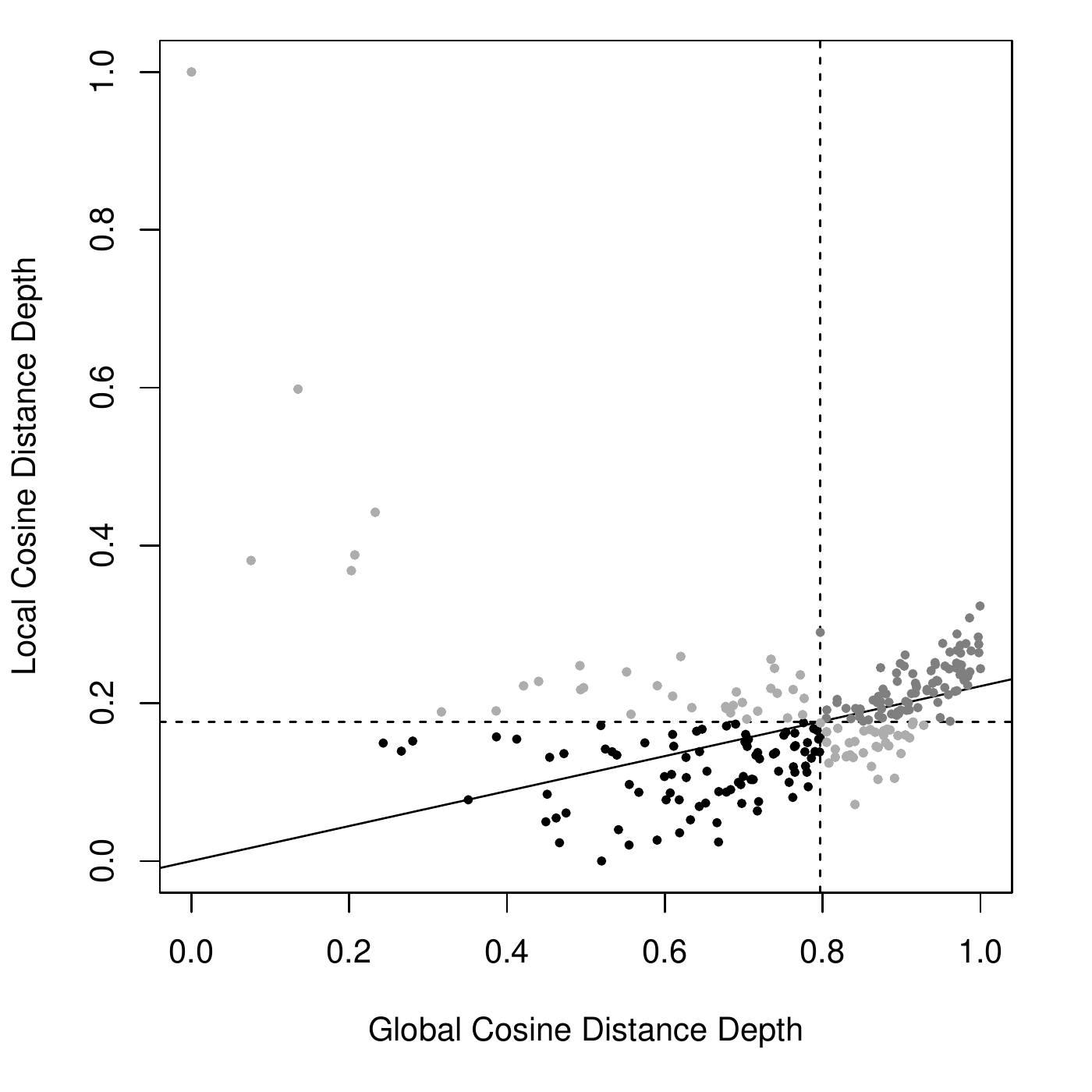}\label{fig:c}}%
 \caption{GLD-plots of a von Mises-Fisher distribution in 5 dimensions with concentration parameter $\kappa = 5$ using the normalized global and local cosine depths. $\mathit{\delta}_{\cos}$ is set equal to $1$ (a), $0.5$ (b) and $0.29$ (c).}%
 \label{firstgld}%
\end{figure}

In the case of more concentrated data (i.e. $\kappa = 20$), the corresponding GLD-plots clearly indicate unimodality for each considered value of $\mathit{\delta}_{\cos}$. Note that all the data points are on the 45-degree line for $\mathit{\delta}_{\cos} = 1$.
\begin{figure}[ht]%
 \centering
 \subfloat[]{\includegraphics[width=0.40\textwidth]{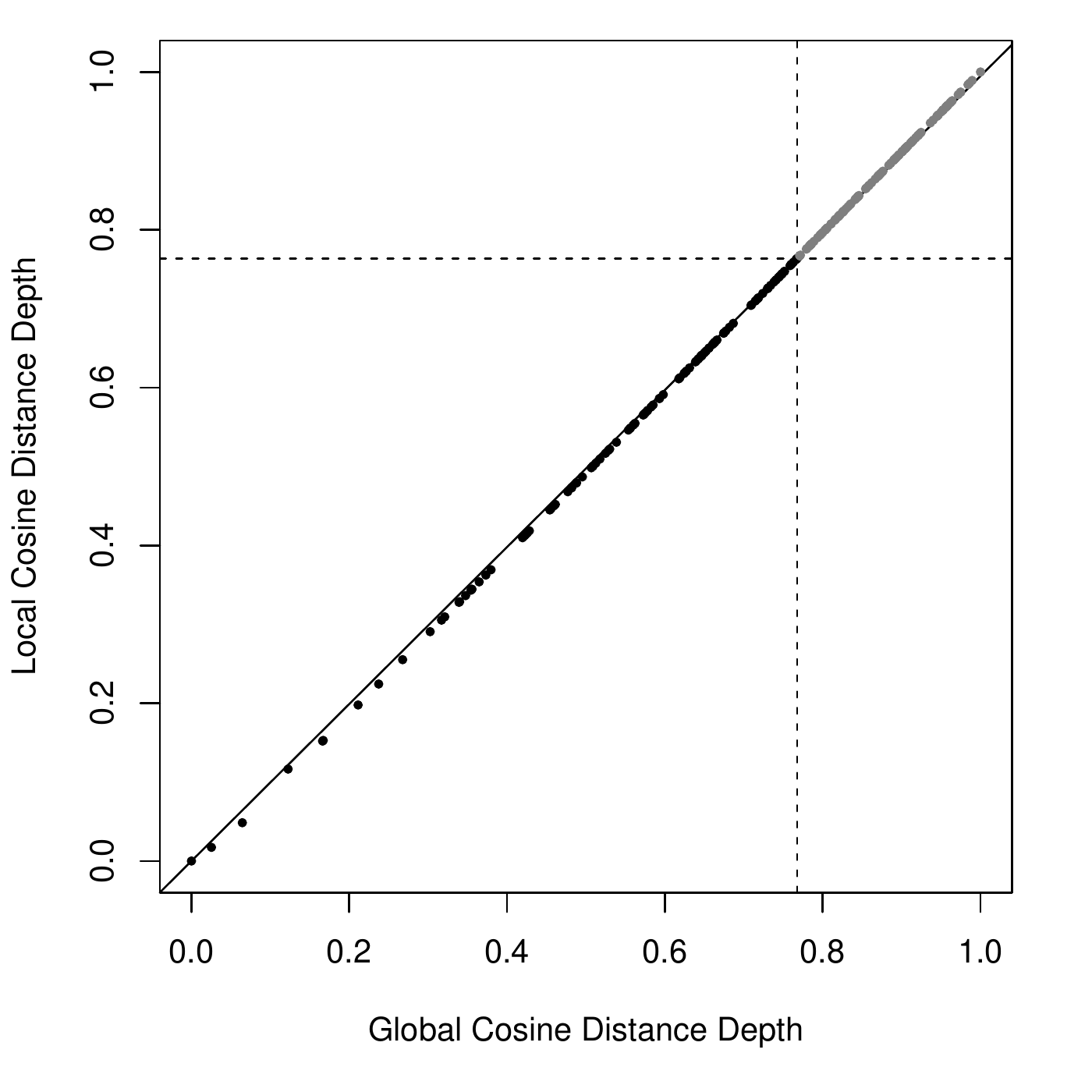}\label{fig:a}}%
 \subfloat[]{\includegraphics[width=0.40\textwidth]{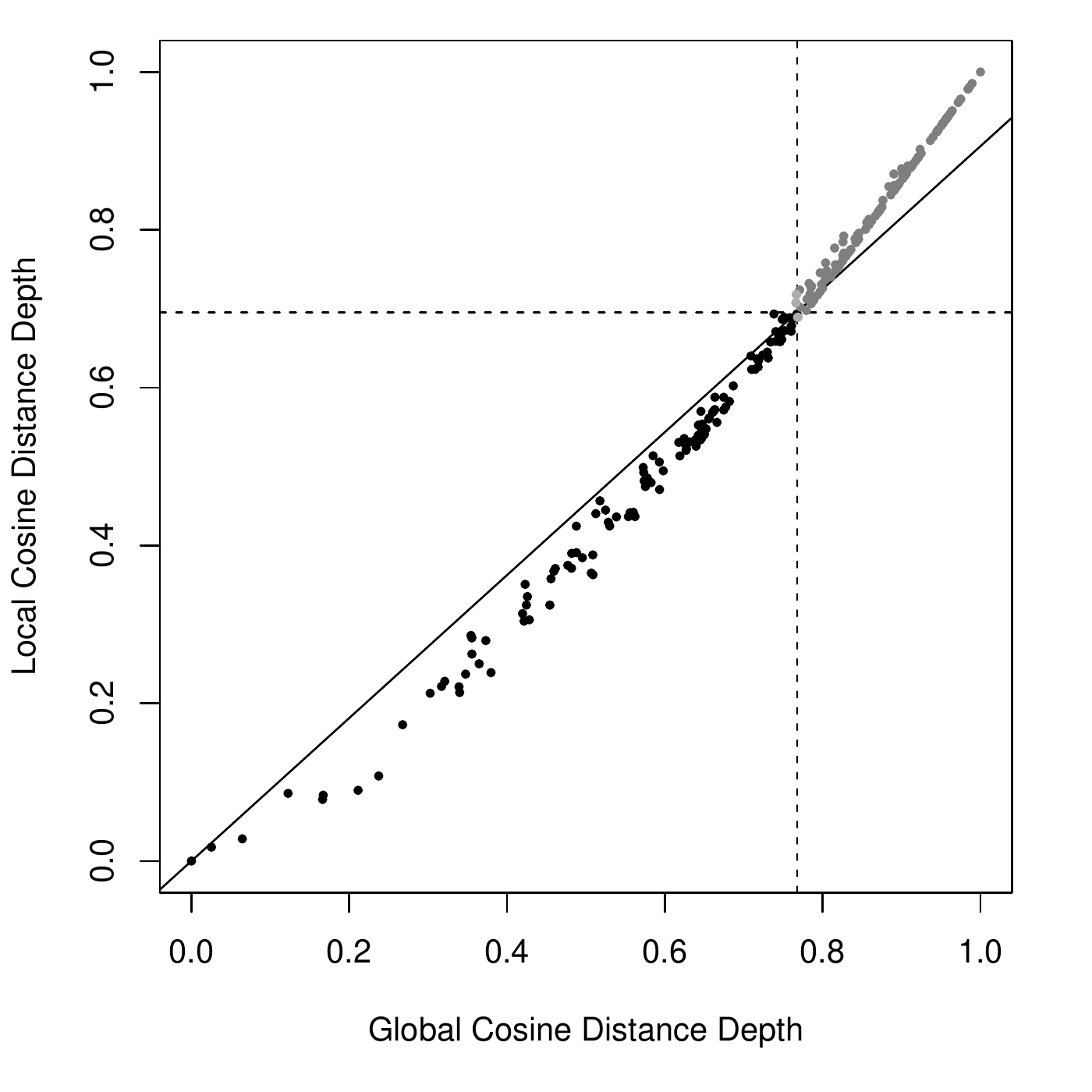}\label{fig:b}}\\
 \subfloat[]{\includegraphics[width=0.40\textwidth]{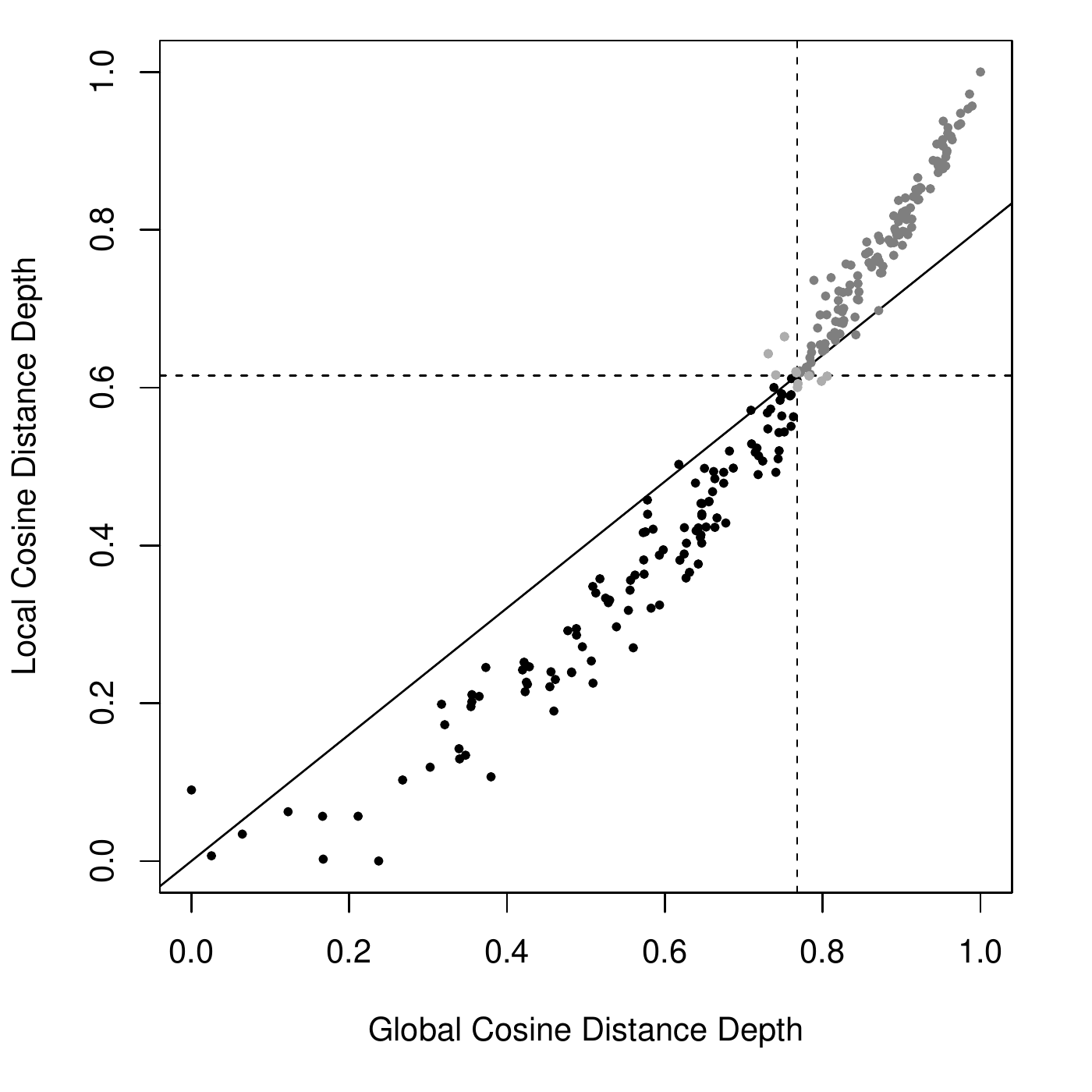}\label{fig:c}}%
 \caption{GLD-plots of a von Mises-Fisher distribution in 5 dimensions with concentration parameter $\kappa = 20$ using the normalized global and local cosine depths. $\mathit{\delta}_{\cos}$ is set equal to $1$ (a), $0.5$ (b) and $0.29$ (c).}%
 \label{secondgld}%
\end{figure}

Then, two bimodal samples were generated according to a mixture of von Mises-Fisher distributions in 5 dimensions with $\mu_1 = (1,0,0,0,0)$ and $\mu_2 = (0,0,0,0,1)$ and concentration parameters $\kappa_{1}=\kappa_{2}=20$. The first was an equally-weighted mixture, while the second had $80 \%$ of the weight on the first component.  
That is, the first distribution has two modes of the same height, the second has a main mode and a well separated minor mode elsewhere on the hypersphere. The corresponding GLD-plots are reported in Figures \ref{thirdgld} and \ref{fourthgld}.
In both cases, points are scattered around, and the $50\%$ deepest sample points according to the global and local depth functions do not clearly lie on the upper-right quadrant, signaling a departure from unimodality for each considered value of $\mathit{\delta}_{\cos}$. It is worth noticing that data points on the GLD-plots depicted in Figure \ref{fourthgld} appear to be U-shaped. It may be due to a great difference between local and global depth rankings with regard to the main mode on the distribution.

\begin{figure}[ht]%
 \centering
 \subfloat[]{\includegraphics[width=0.40\textwidth]{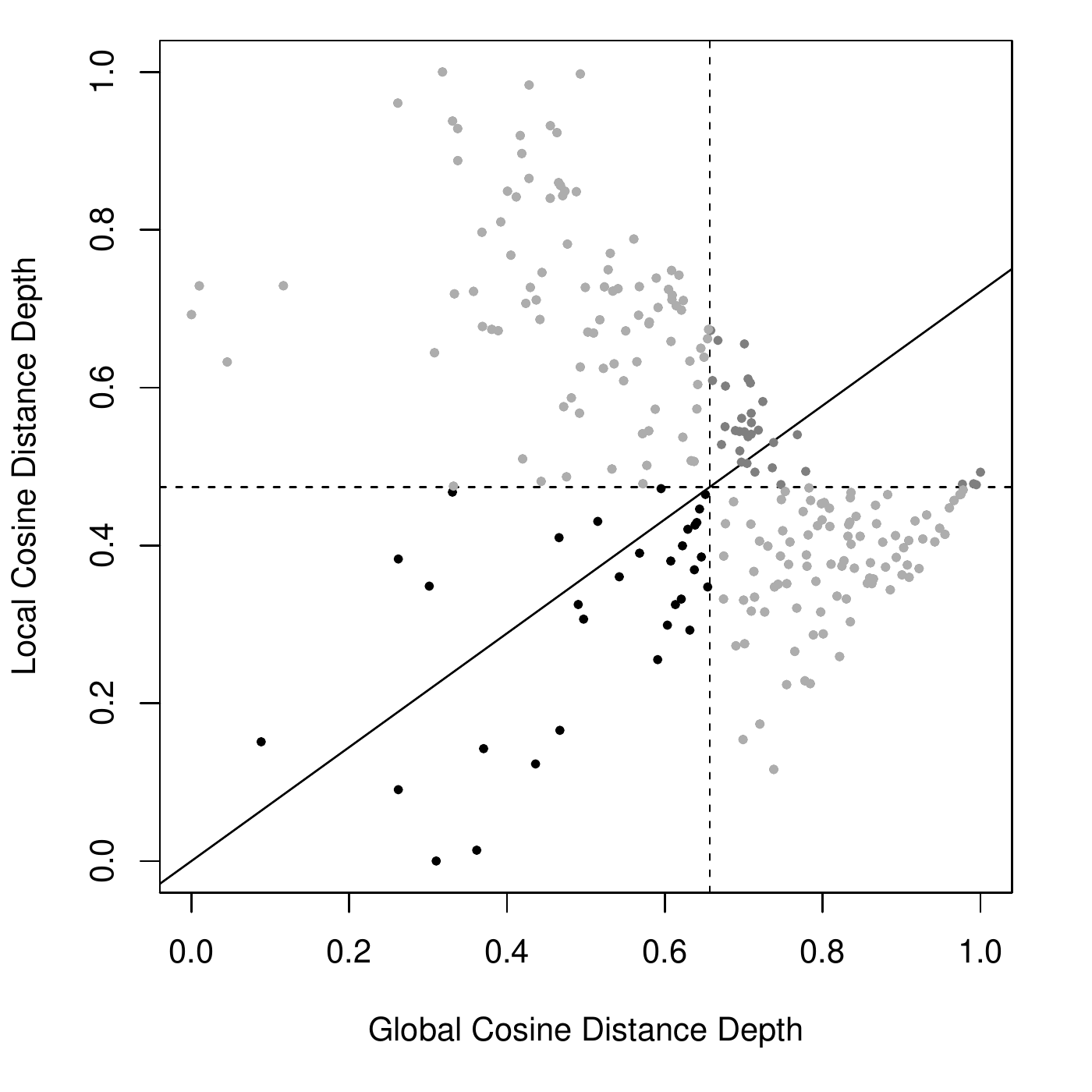}\label{fig:a}}%
 \subfloat[]{\includegraphics[width=0.40\textwidth]{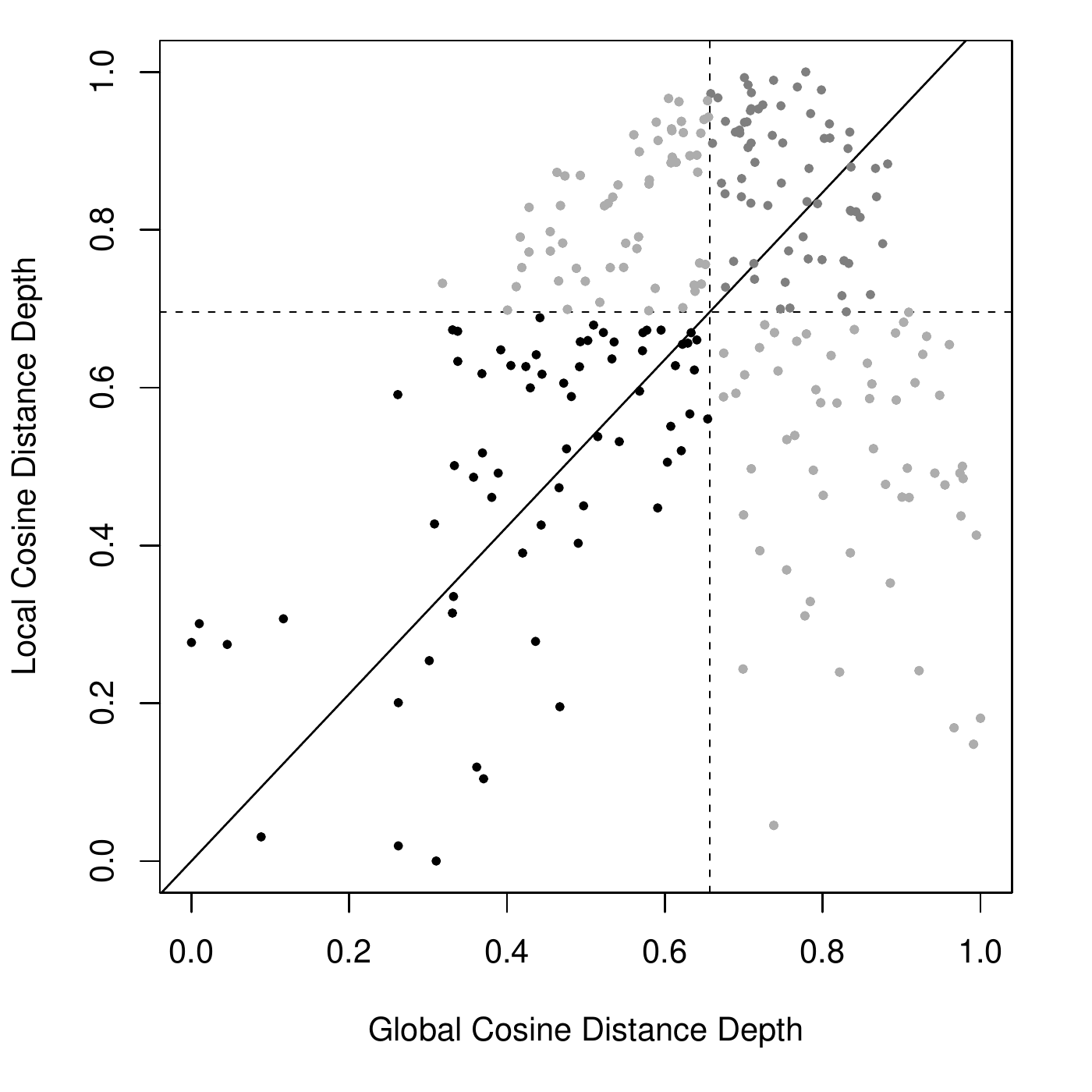}\label{fig:b}}\\
 \subfloat[]{\includegraphics[width=0.40\textwidth]{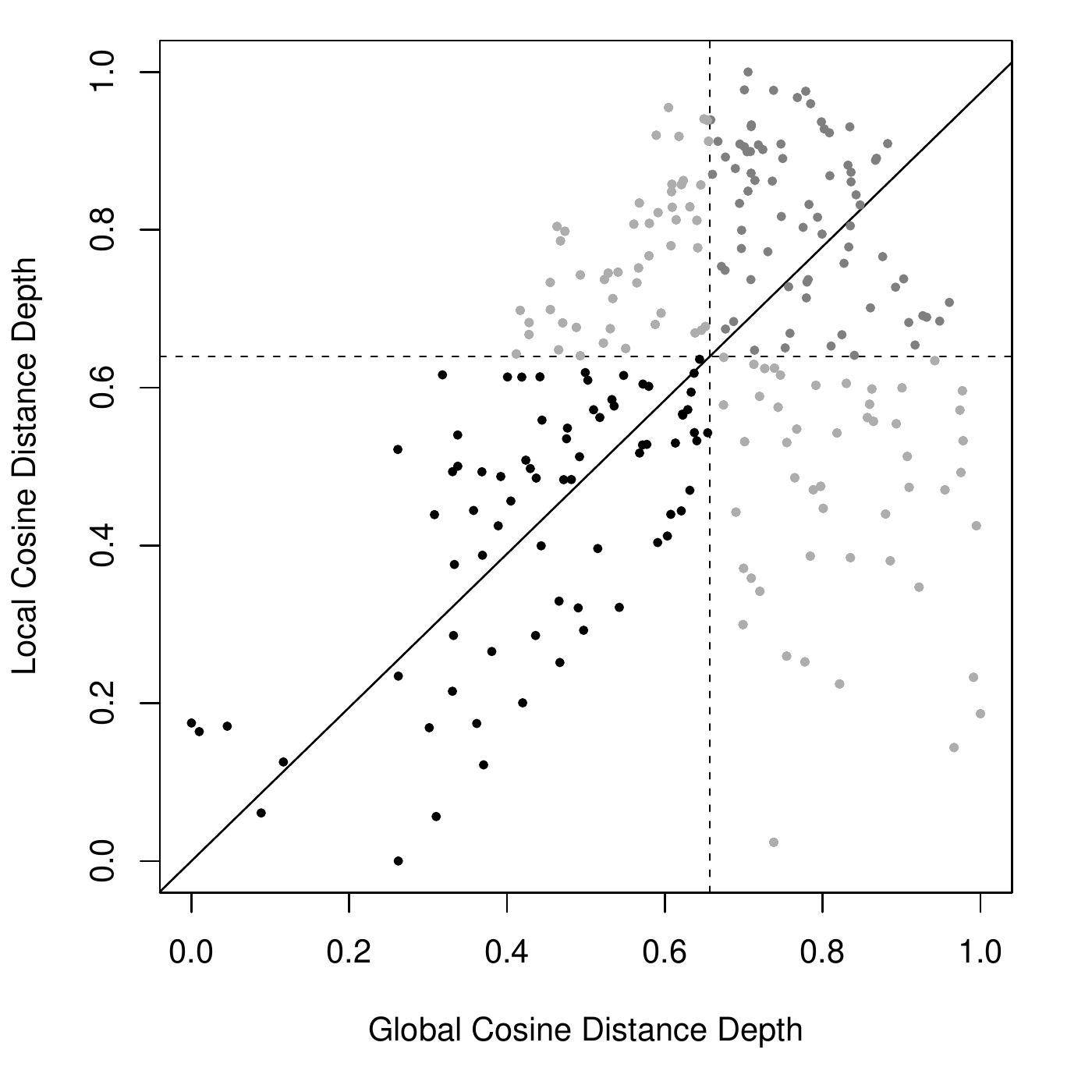}\label{fig:c}}%
 \caption{GLD-plots of a bimodal sample generated from an equally-weighted von Mises-Fisher distribution in 5 dimensions using the normalized global and local cosine depths. $\mathit{\delta}_{\cos}$ is set equal to $1$ (a), $0.5$ (b) and $0.29$ (c).}%
 \label{thirdgld}%
\end{figure}

\begin{figure}[ht]%
 \centering
 \subfloat[]{\includegraphics[width=0.40\textwidth]{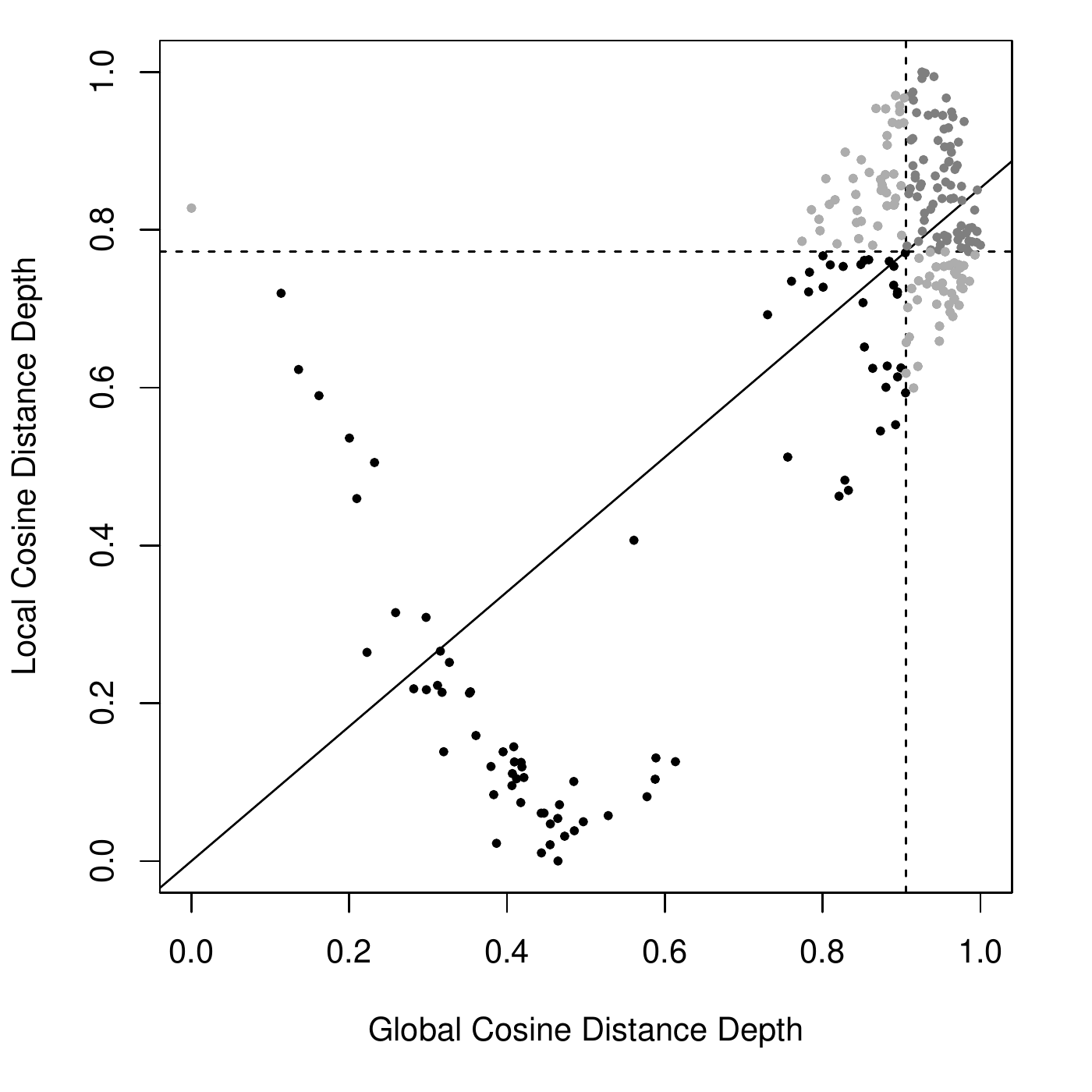}\label{fig:a}}%
 \subfloat[]{\includegraphics[width=0.40\textwidth]{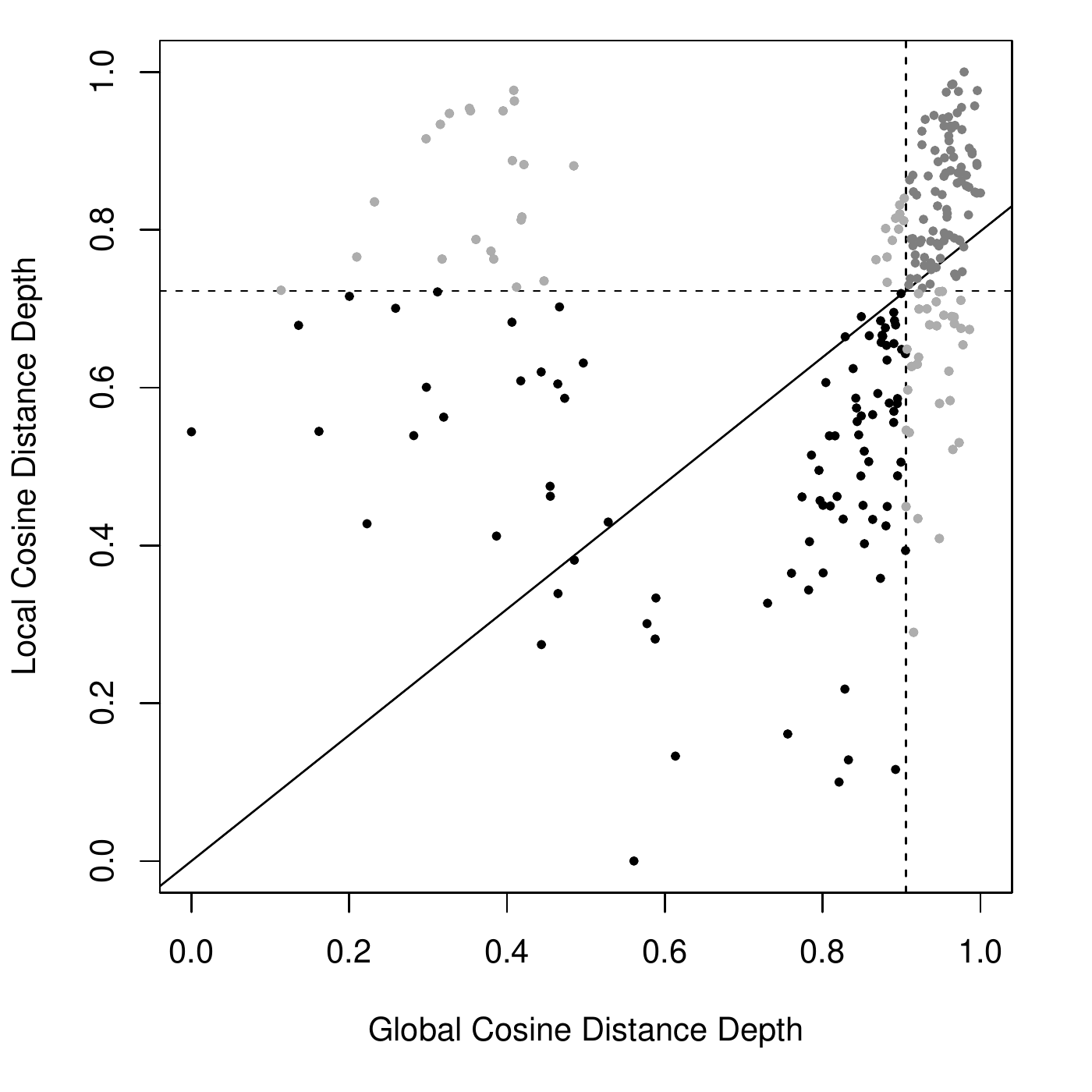}\label{fig:b}}\\
 \subfloat[]{\includegraphics[width=0.40\textwidth]{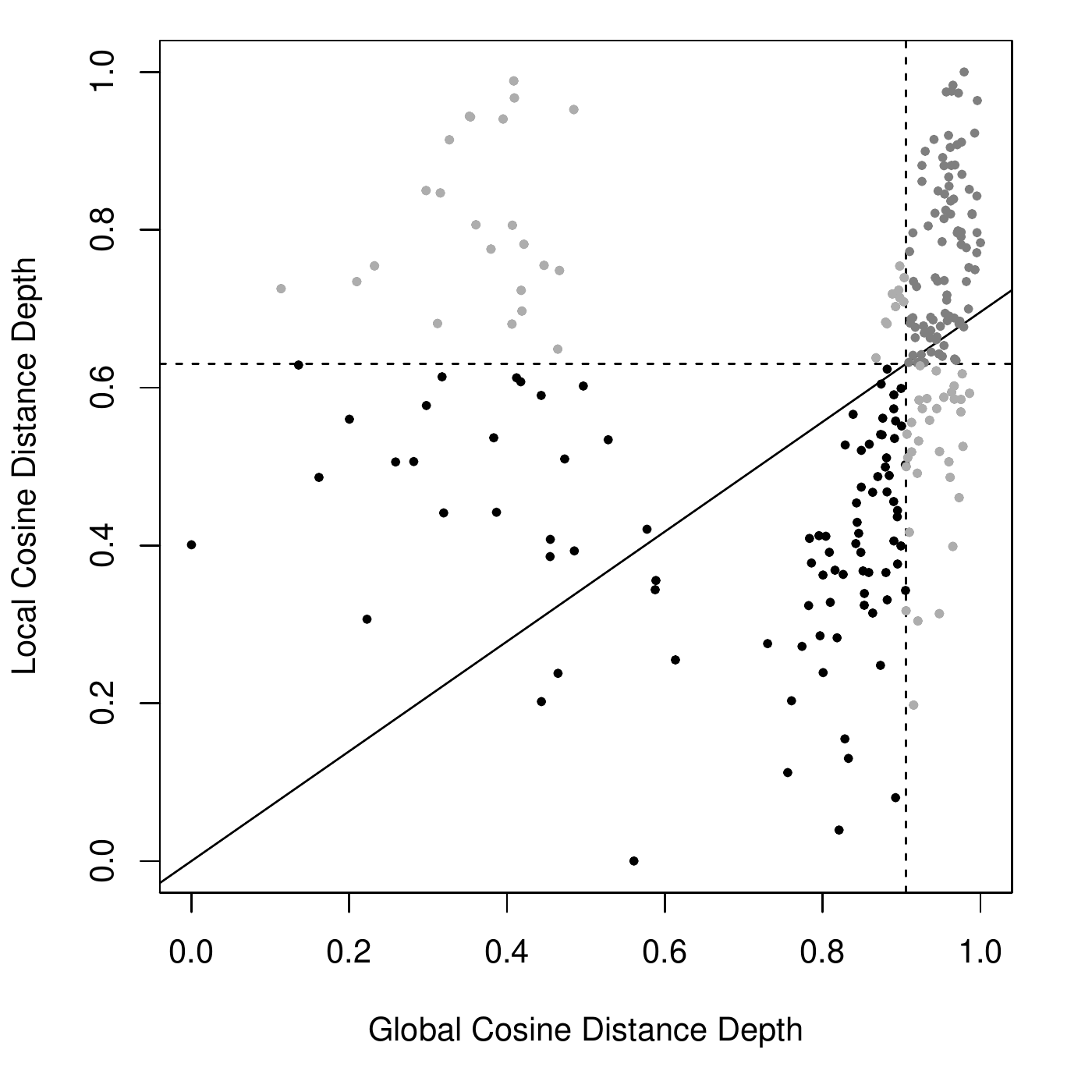}\label{fig:c}}%
 \caption{GLD-plots of a bimodal sample generated from a mixture of von Mises-Fisher distributions $80 \%$ of the weight on the first component in 5 dimensions using the normalized global and local cosine depths. $\mathit{\delta}_{\cos}$ is set equal to $1$ (a), $0.5$ (b) and $0.29$ (c).}%
 \label{fourthgld}%
\end{figure}

\section{Real data examples}
\label{sec:RealDataExample}
 
In the following two real data sets on $S^{2}$ are considered, so that they can be also viewed in their original sample space to facilitate the evaluation of the proposed tool.

\subsection{Directions of flights}

The sample refers to the $576$ flights departed from Brussels National Airport on March 3, 2014. Data (expressed in latitude and longitude coordinates) are centered and their contour plot is reported in Figure \ref{gldflights} (a). One can see such data show a major mode and two minor clusters. The corresponding GLD-plots for $\mathit{\delta}_{\cos} = 0.29$, $0.5$ and $1$ are depicted Figure \ref{gldflights} (b), (c) and (d), respectively. In all such plots, the data points in the upper-right quadrant suggest that the sample has a main mode, while the presence of some  irregularities on the left side of the plot, especially for $\mathit{\delta}_{\cos} = 0.29$ and $0.5$, indicates that data someway deviates from unimodality. 

\begin{figure}[ht]%
 \centering
 \subfloat[]{\includegraphics[width=0.40\textwidth]{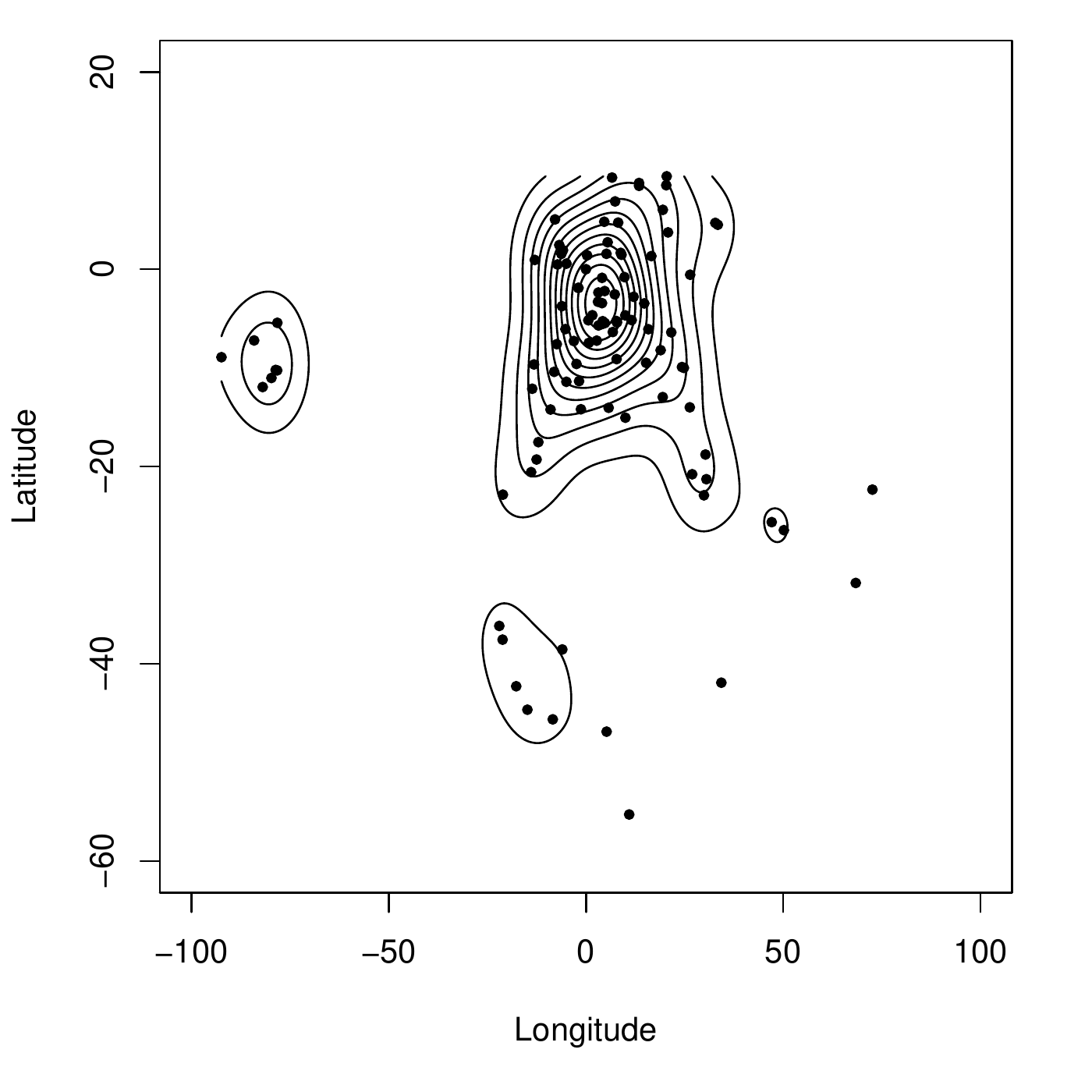}\label{fig:a}}%
 \subfloat[]{\includegraphics[width=0.40\textwidth]{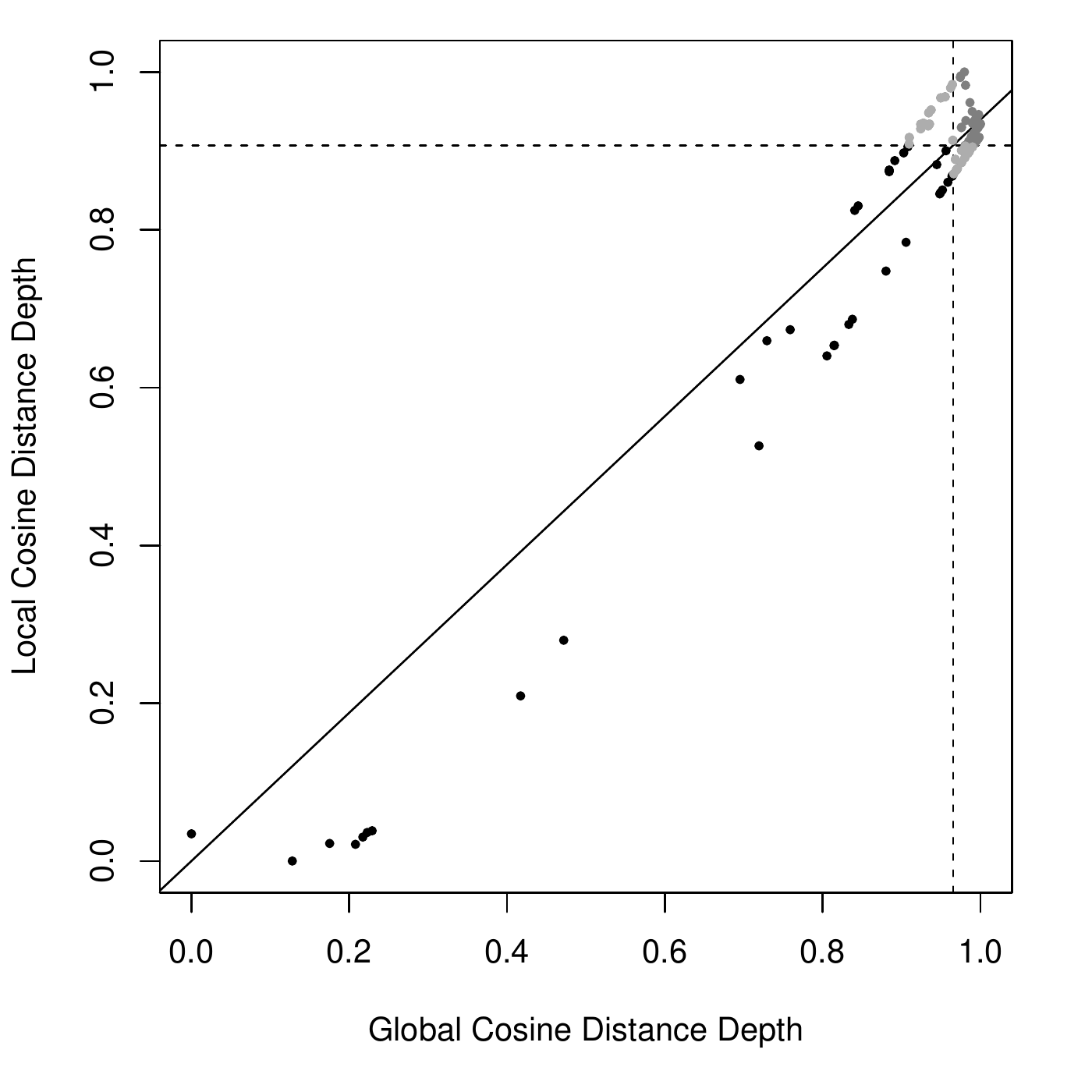}\label{fig:b}}\\
 \subfloat[]{\includegraphics[width=0.40\textwidth]{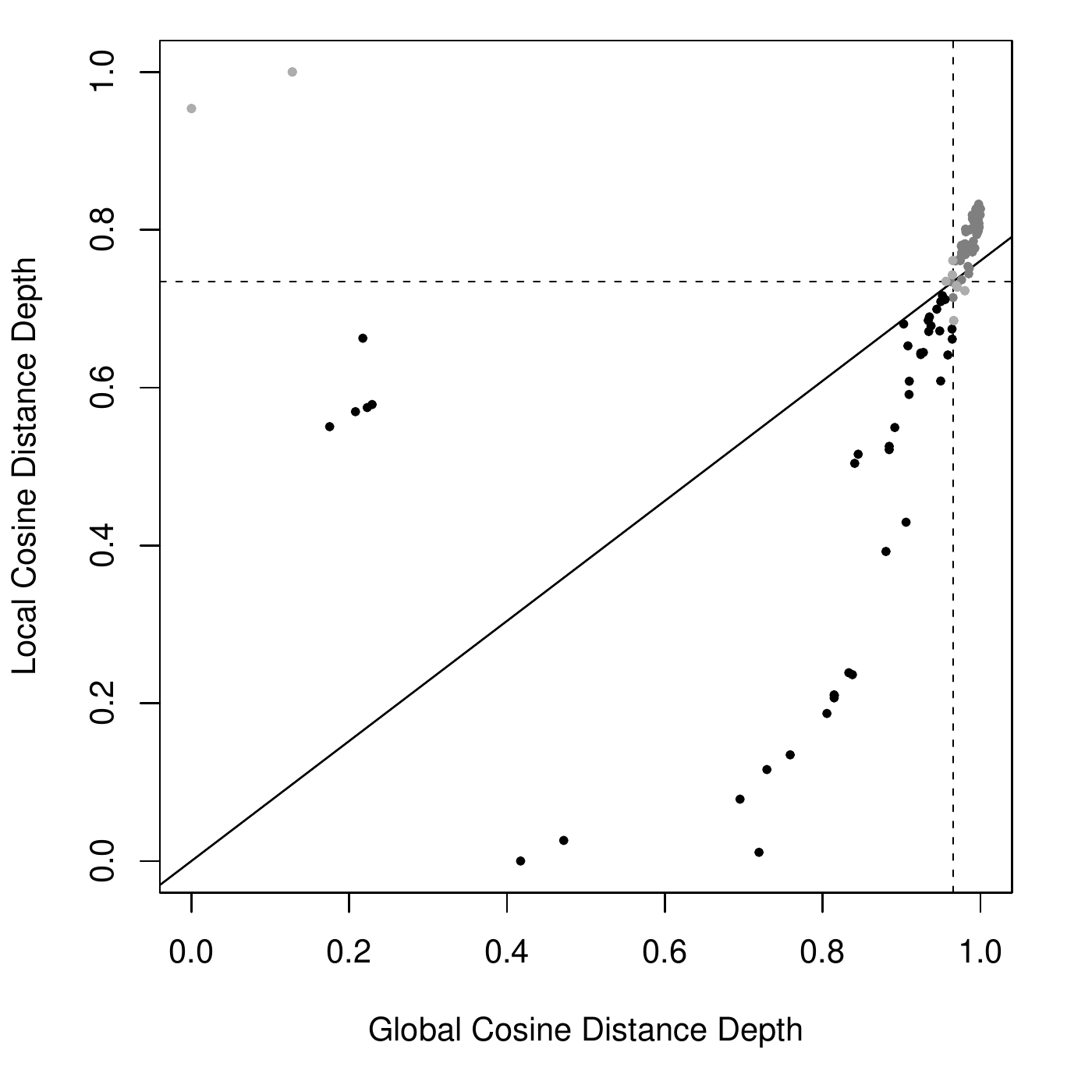}\label{fig:c}}%
 \subfloat[]{\includegraphics[width=0.40\textwidth]{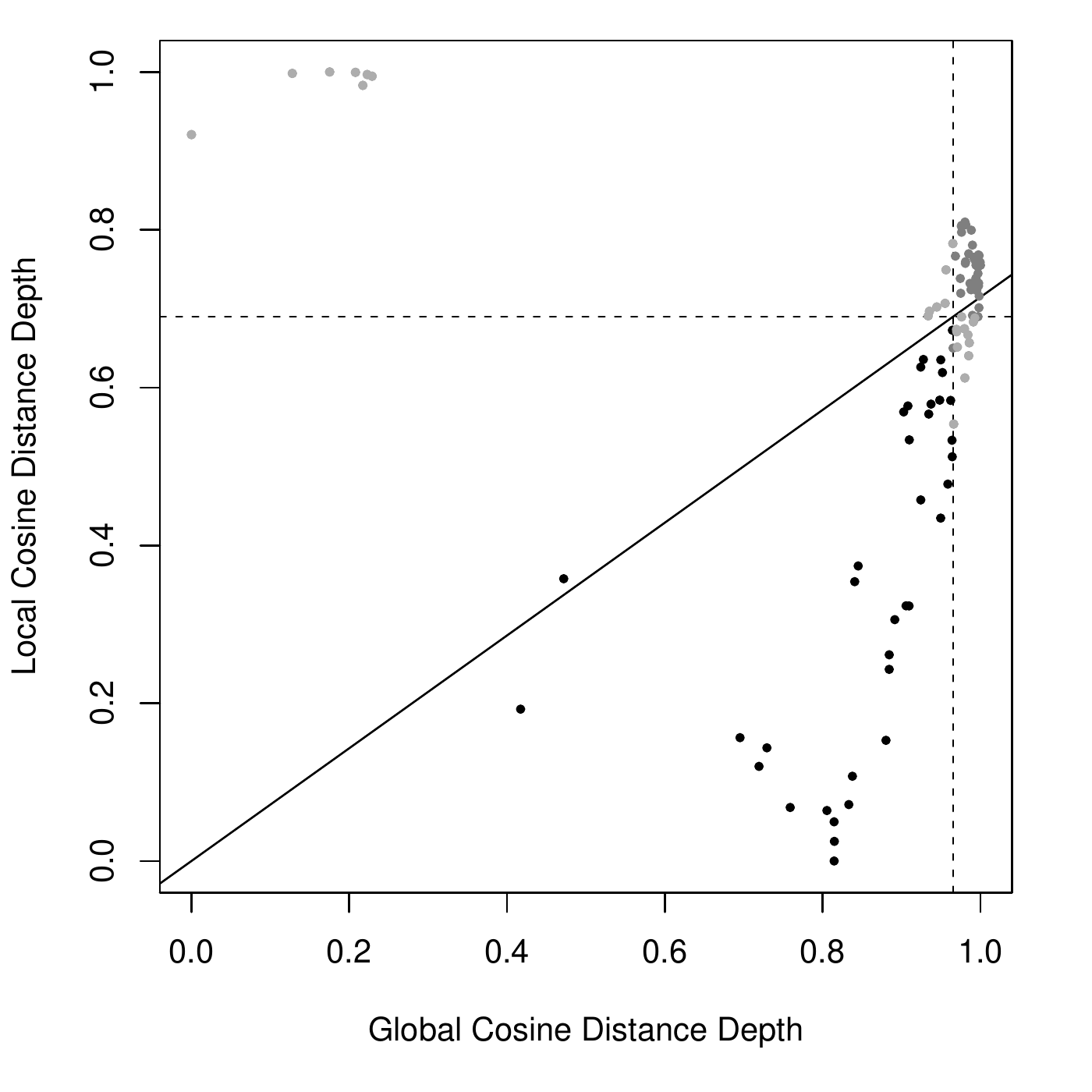}\label{fig:c}}%
 \caption{Contour plot of directions of the $576$ flights departing from Brussels National Airport on March 3, 2014 (a), and the corresponding GLD-plots. $\mathit{\delta}_{\cos}$ is set equal to $1$ (b), $0.5$ (c) and $0.29$ (d).}%
 \label{gldflights}%
\end{figure}

\subsection{Magnetic remanence measurements}
 
Data concern $107$ measurements of magnetic remanence in specimens of Precambrian volcanics. Indeed, magnetic field and magnetization directions can be treated as unit vectors anchored at the center of a unit sphere (see \citealp{schmidt1985} and \citealp{fisher1993statistical}). The contour plot of the data (expressed in latitude and longitude coordinates), reported in Figure \ref{mag} (a), clearly indicates two main modes. Also here the GLD-plots (for all the considered values of $\mathit{\delta}_{\cos}$) suggest a deviation from unimodality, with the data points which appear to be U-shaped as it occurred for simulated bimodal data presented in Section \ref{sec:exploratory}. 

\begin{figure}[ht]%
 \centering
 \subfloat[]{\includegraphics[width=0.40\textwidth]{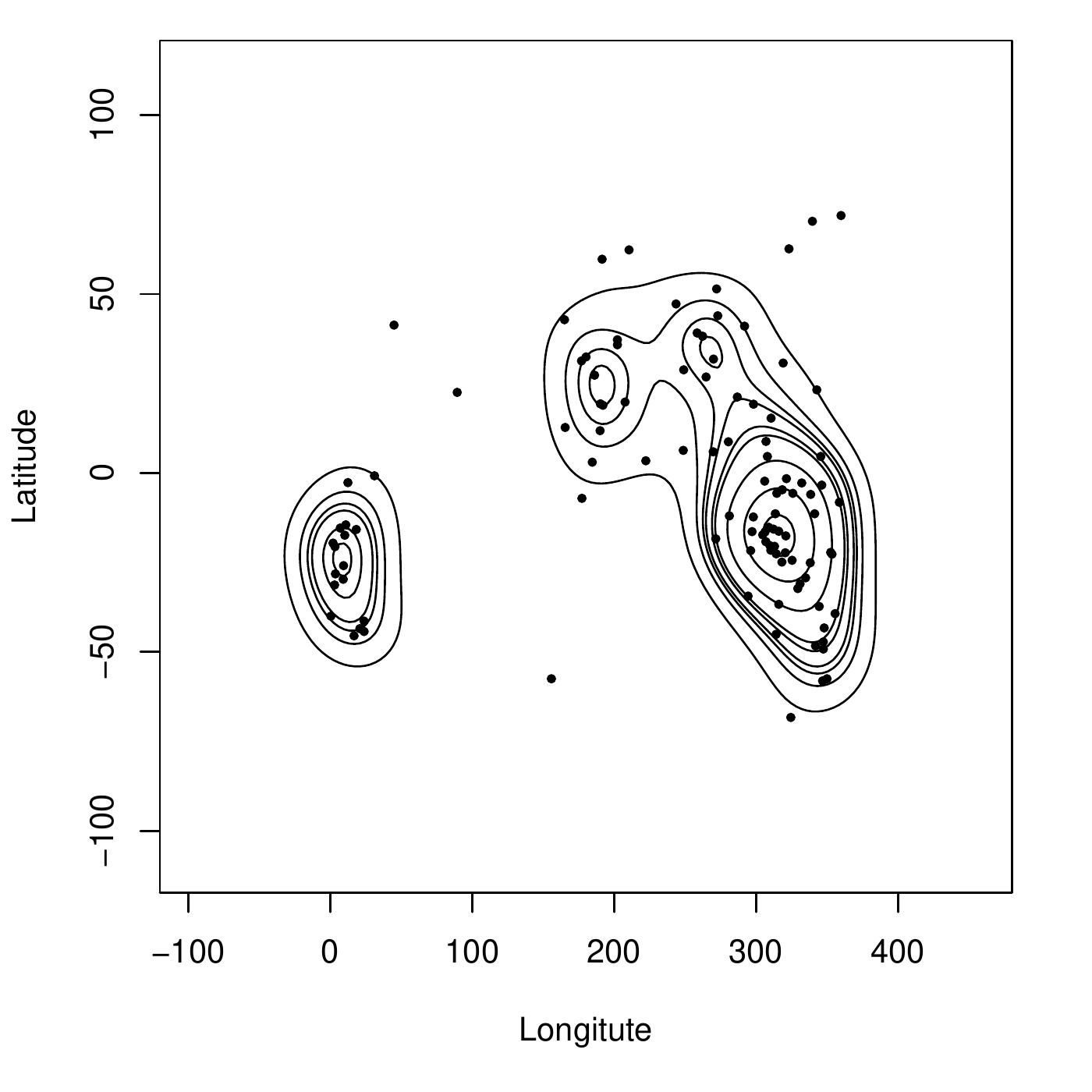}\label{fig:a}}%
 \subfloat[]{\includegraphics[width=0.40\textwidth]{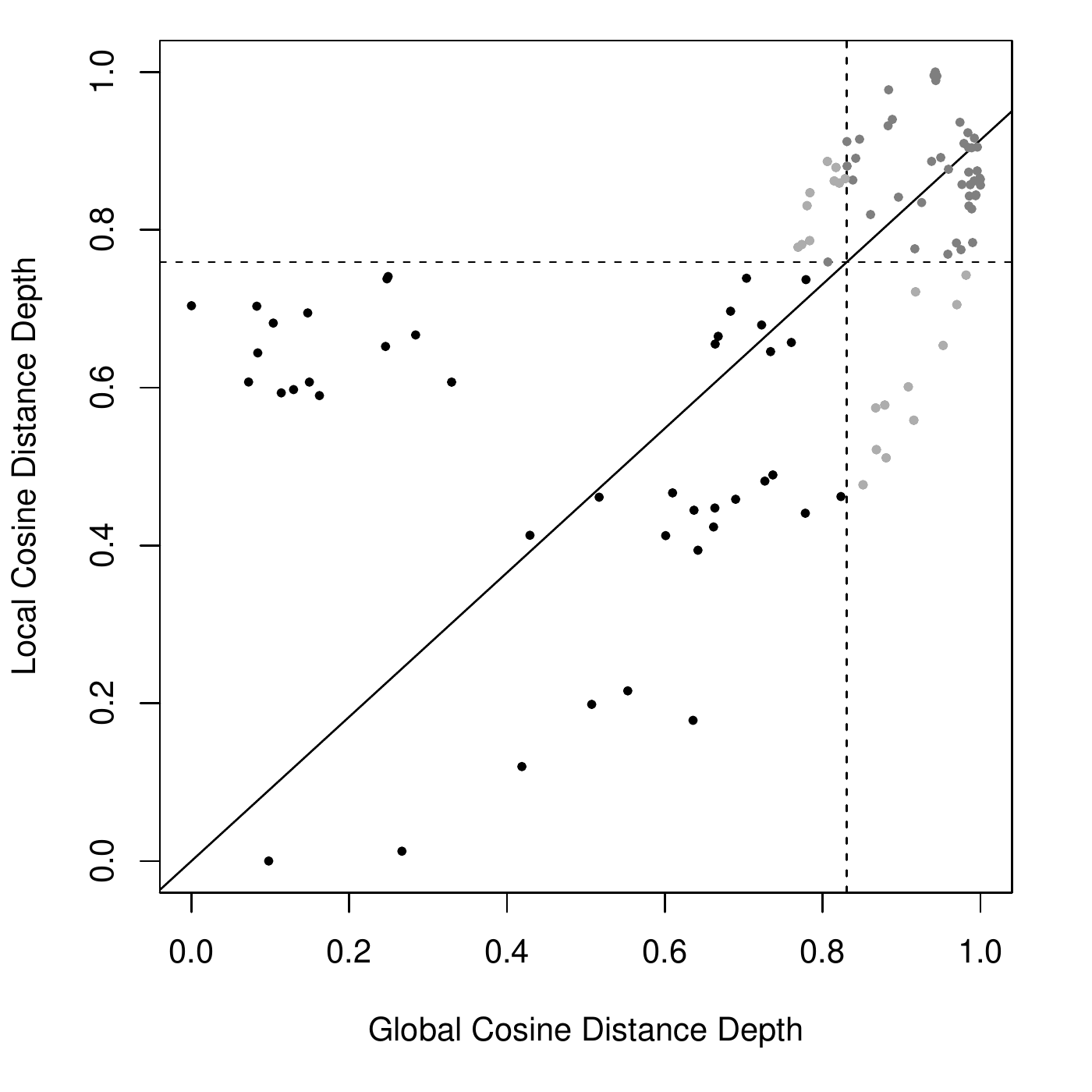}\label{fig:b}}\\
 \subfloat[]{\includegraphics[width=0.40\textwidth]{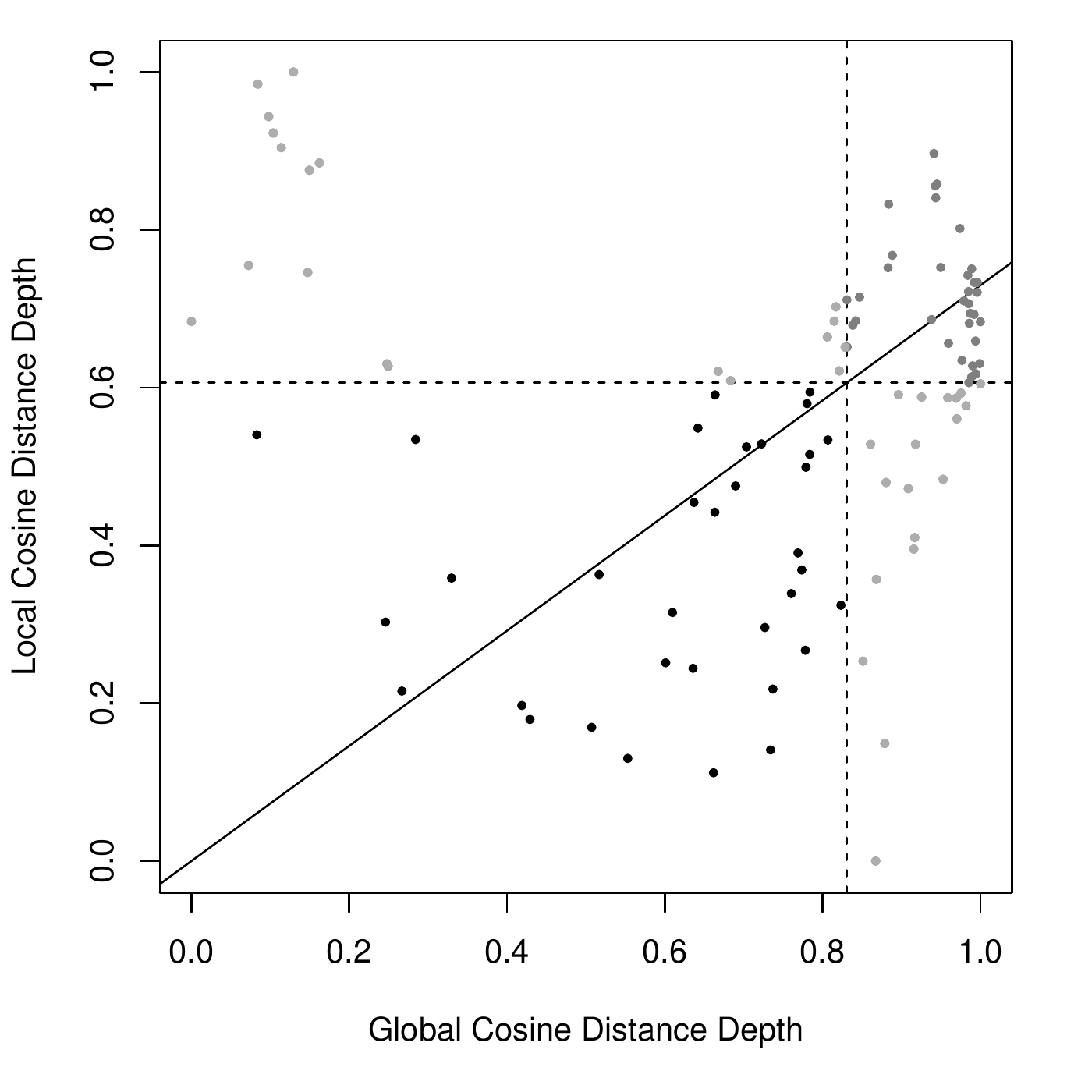}\label{fig:c}}%
 \subfloat[]{\includegraphics[width=0.40\textwidth]{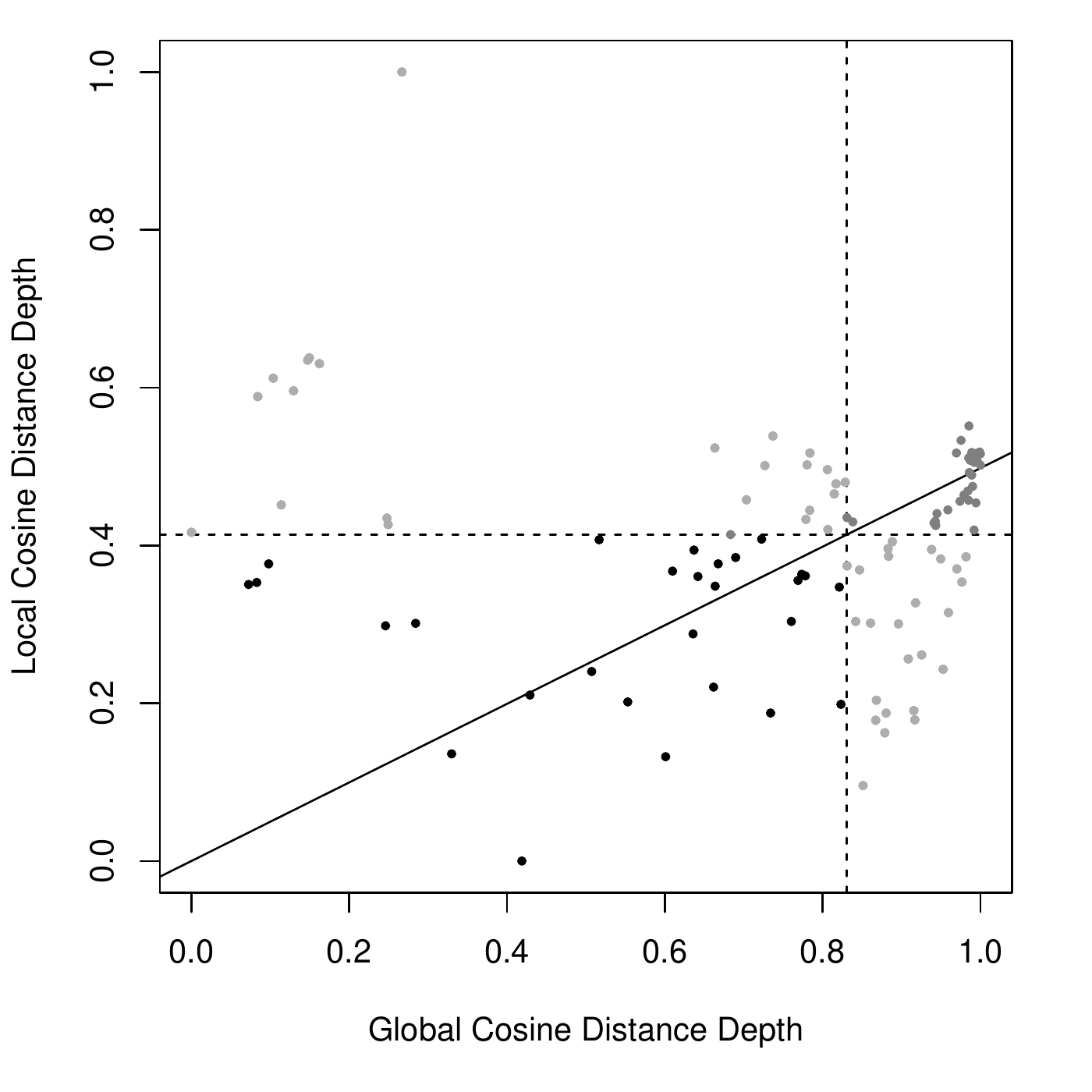}\label{fig:c}}%
 \caption{Contour plot of 107 measurements of magnetic remanence in specimens of Precambrian volcanics (a), and the corresponding GLD-plots. $\mathit{\delta}_{\cos}$ is set equal to $1$ (b), $0.5$ (c) and $0.29$ (d).}%
 \label{mag}%
\end{figure}

\section{Concluding remarks}
\label{sec:conc}

Depth function is a useful tool in nonparametric statistics which is aimed at providing an important contribution to directional data analysis too. In this work, local distance-based depth functions for directional data are presented and exploited in order to propose a graphical tool,  the Global vs Local depth plot (GLD-plot), for investigating unimodality on high dimensional spherical spaces. The proposal allows for the comparison between local and global depth-induced rankings for a given sample by means of an easy interpretable two-dimensional scatterplot. The proposed tool was evaluated through simulated examples and two real data sets, and it seems to merit particular interest. The GLD-plot does not aim to be an inference method for making automated decisions about a distribution. Rather, it is to be considered a useful tool to support the data analyst. 

\clearpage

\bibliographystyle{apalike}
\bibliography{references}

\end{document}